%% file: main.tex
\documentclass[11pt]{article}

\newtheorem{mythm}{\bf Theorem}
\newtheorem{myprob}{\bf Problem}

\newtheorem{mycol}{\bf Corollary}

\newtheorem{remark}{\bf Remark}

\newtheorem{assumption}{\bf Assumption}

\usepackage[linesnumbered,ruled,vlined]{algorithm2e}
\usepackage{tikz}

\usepackage{subfig}
\usepackage{graphicx}
\graphicspath{{./Pics/}}
\DeclareGraphicsExtensions{.pdf}
\usepackage{amsmath}
\usepackage{amsfonts}
\usepackage{arydshln}
\usepackage{algpseudocode}
\usepackage{verbatim}
\usepackage{enumerate}
\usepackage{rotating}
\usepackage{color}
\usepackage{caption} 
\usepackage{mathrsfs}
\usepackage[linesnumbered,ruled,vlined]{algorithm2e}
\usepackage{float}
\usepackage{hyperref}
\hypersetup{hidelinks} 
\usepackage{soul}
\usepackage{booktabs}
\usepackage{multirow}
\usepackage{makecell}
\usepackage{array}
\usepackage{tcolorbox}
\usepackage{fancyhdr}
\usepackage{framed}
\usepackage[para,online,flushleft]{threeparttable}
\usepackage{flushend}
\usepackage{cite}
\usepackage{mathtools}
\usepackage[margin=1in]{geometry}

\newcommand{\proof}[1]{\par\noindent\textbf{Proof.} #1 \hfill$\square$\par}

\title{\LARGE Signal Temporal Logic Control Synthesis among \\ Uncontrollable Dynamic Agents with Conformal Prediction}

\author{Xinyi Yu, Yiqi Zhao, Xiang Yin, and Lars Lindemann
	\thanks{This work was supported in part by Annenberg Fellowship, Viterbi School of Engineering Fellowship, National Science Foundation through the grant IIS-SLES-2417075, and National Natural Science Foundation of China 62173226. 
    Xinyi Yu, Yiqi Zhao, and Lars Lindemann are with Thomas Lord Department of Computer Science, University of Southern California, Los Angeles, CA 90089, USA.
	e-mail: {\tt\small $\{$xinyi.yu12, yiqizhao, llindema$\}$@usc.edu}.
    Xiang Yin is with School of Automation and Sensing, Shanghai Jiao Tong University, Shanghai 200240, China.
	e-mail: {\tt\small yinxiang@sjtu.edu.cn}.}
}
\date{}

\begin{document}

    



\maketitle
\thispagestyle{empty}
\pagestyle{plain}

\begin{abstract}
    The control of dynamical systems under temporal logic specifications among uncontrollable dynamic agents is challenging due to the agents' a-priori unknown behavior. Existing works have considered the problem where either all agents are controllable, the agent models are deterministic and known, or no safety guarantees are provided. We propose a predictive control synthesis framework that guarantees, with high probability, the satisfaction of signal temporal logic (STL) tasks that are  defined over a controllable system in the presence of  uncontrollable stochastic agents. We use trajectory predictors and conformal prediction to construct probabilistic prediction regions for each uncontrollable agent that are valid over multiple future time steps. Specifically, we construct a normalized prediction region over all agents and time steps to reduce conservatism and increase data efficiency. We then formulate a worst-case  bilevel mixed  integer program (MIP) that accounts for all agent realizations within the prediction region to obtain an open-loop controller that provably guarantee task satisfaction with high probability. To efficiently solve this bilevel MIP, we propose an equivalent MIP program based on KKT conditions of the original bilevel formulation. Building upon this, we design a closed-loop controller, where both recursive feasibility and task satisfaction can be guaranteed with high probability. We illustrate our control synthesis framework on two case studies.
\end{abstract}

\section{Introduction}

Consider the following scenarios in which autonomous dynamical systems need to operate safely among uncontrollable dynamical agents: mobile service robots operating in pedestrian-rich environments, self-driving cars navigating through dense urban traffic, and leader-follower systems such as truck platoons reducing air drag to improve fuel economy. All these scenarios have in common that the autonomous system needs to satisfy complex safety-critical specifications in the presence of uncontrollable agents. Achieving these specifications is challenged by the fact that communication is prohibited, limited and unstable, or only one-directional,  while the agents' intentions may not always be known. In this setting, the autonomous system needs to make predictions about the behavior of the dynamic agents and select safe actions accordingly. In this work, we propose a predictive control synthesis framework with probabilistic safety for complex temporal logic specifications.

Control of multi-agent systems under temporal logic specifications has attracted much attention, see e.g., \cite{guo2015multi,kloetzer2009automatic,kantaros2020stylus,sun2022multi,kantaros2019temporal,buyukkoccak2021distributed,lindemann2019control,srinivasan2018control}. However, the agent dynamics here are usually known and agents are assumed to be collaborative. In this paper, we are interested in control synthesis among uncontrollable dynamic stochastic agents under temporal logic tasks, and specifically in settings where communication is limited and no information about agent intentions and  models is available. This problem has been studied in the literature, e.g., using sampling-based approaches \cite{kalluraya2023multi, hoxha2016planning}, game theory \cite{li2021safe}, supervisory control theory \cite{kloetzer2012ltl}, or optimization-based methods \cite{ulusoy2014receding,ayala2011temporal,purohit2021dt}. These works, however, only consider qualitative temporal logic specifications such as linear temporal logic, while we are here interested in quantitative temporal logics that allow to measure the quality of task satisfaction. More importantly, these works make simplifying assumptions on the models of the agents  such as being linear or having Gaussian noise distributions. Recent works in \cite{lindemann2023safe,dixit2023adaptive,tonkensscalable,muthali2023multi} make no such assumptions and instead use learning-enabled trajectory predictors and conformal prediction  to obtain probabilistic prediction regions that can be used for downstream control. However, these works do not consider complex specifications expressed in temporal logics. 
Note that extending the control frameworks aforementioned to the temporal logic case is not straightforward, 
since temporal logic specifications are defined over agent trajectories while the aforementioned works only consider point-wise state constraints (see more discussions at the end of section \ref{subsec:prob}).

\textbf{Contribution. } We consider the signal temporal logic (STL) control synthesis problem in the presence of uncontrollable dynamic agents without making any assumptions on the agents' dynamics and distribution.  Specifically, we make the following contributions:
\begin{itemize}
    \item We use trajectory predictors and conformal prediction to construct probabilistic prediction regions that contain the trajectories of multiple uncontrollable agents. \vspace{-6pt}
    \item We propose a new mixed integer program (MIP) based on a worst case formulation of the STL constraints over all uncontrollable agent trajectories within these prediction regions. We present an equivalent MIP program based on KKT conditions of the original MIP to efficiently compute the open-loop controller that provably guarantee task satisfaction with high probability. \vspace{-6pt}
    \item To obtain a closed-loop controller, we introduce a  recursively feasible model predictive control (MPC) framework, where  recursive feasibility through the task  horizon can be guaranteed with high probability.
    We again guarantee task satisfaction with high probability. \vspace{-6pt}
    \item We  validate and illustrate the efficacy of our algorithms on temperature control and robot motion planning problems.
\end{itemize}

 \textbf{Organization. } After discussing  related work in Section \ref{sec:related}, we present preliminaries and the problem definition in Section \ref{sec:prob}. We show how to compute prediction regions for multiple uncontrollable agents using conformal prediction in Section~\ref{sec:cp} and present predictive STL open-loop and closed-loop control synthesis frameworks in Section~\ref{sec:solution_open} and \ref{sec:solution_closed}, respectively. We present simulations in Section \ref{sec:case} and conclude in Section \ref{sec:con}.

\subsection{Related Work}
\label{sec:related}

Quantitative temporal logics such as signal temporal logic (STL) \cite{maler2004monitoring} have increasingly been used in control. The reason for their success in control is their ability to reason over the robustness of a control system with respect to the specification via the STL quantitative semantics \cite{fainekos2009robustness,donze2010robust}. The quantitative semantics of STL provide a natural control objective that can be maximized, e.g., via gradient-based optimization \cite{pant2017smooth,pant2018fly,mehdipour2019arithmetic, gilpin2020smooth}, mixed integer programming \cite{raman2014model, sadraddini2018formal,kurtz2022mixed}, or control barrier functions \cite{lindemann2018control, charitidou2022receding, xiao2021high}.

While STL control synthesis problems have been well studied for deterministic systems, the problem becomes more challenging for stochastic systems. There are various approaches that ensure probabilistic task satisfaction, e.g., via mixed integer programming \cite{farahani2017shrinking,jha2018safe,sadigh2016safe}, sampling-based approaches \cite{ilyes2023stochastic,vasile2017sampling,scher2023ensuring}, and stochastic control barrier functions \cite{jagtap2020formal}. Another recent line of work considers the minimization of mathematical risk of  STL specifications \cite{safaoui2020control, akella2024sample,hashemi2023risk}. In addition, there are also some related approaches for reactive STL control synthesis, see e.g., \cite{lindemann2021reactive, gundana2021event, gundana2022event,raman2015reactive}. 

Conceptually closest to our work are \cite{lindemann2023safe, dixit2023adaptive, muthali2023multi} where the authors apply conformal prediction to obtain probabilistic prediction regions of uncontrollable dynamic stochastic agents. 
Conformal prediction is a lightweight statistical tool for uncertainty quantification, originally introduced  in \cite{shafer2008tutorial,vovk2005algorithmic}, that has recently attracted  attention within the machine learning community, see e.g., \cite{angelopoulos2021gentle,fontana2023conformal} for up-to-date tutorials.
Compared to \cite{lindemann2023safe, dixit2023adaptive, muthali2023multi}, we consider complex signal temporal logic specifications, that are defined over agent trajectories as opposed to considering simple state constraints. 
Previous works hence do not apply to our settings as we have to account for state couplings at different times.
Furthermore, complex temporal logic tasks require the data-efficient computation of non-conservative probabilistic prediction regions  to promote long horizon control with multiple agents.
We therefore propose an efficient method for computing prediction regions, and then formulate a robust control problem which is fundamentally different and less conservative from Lipschitz-based formulations proposed in the aforementioned works.

\section{Preliminaries \& Problem Definition}\label{sec:prob}

We consider an {autonomous system}, e.g., a robotic system or an autonomous car, that is described by the discrete-time dynamics
\begin{align}\label{eq:sys}
	x_{k+1} = f(x_k, u_k), \;\;\; x_0\in\mathcal{X}
\end{align}
where the function $f: \mathcal{X} \times \mathcal{U} \to \mathcal{X}$ is defined over the state and input domains $\mathcal{X}$ and $\mathcal{U}$, respectively. Here, $x_0$ is the known initial state and $x_k \!\in\! \mathcal{X} \subset \mathbb{R}^{n_x}$ and $u_k \!\in\! \mathcal{U} \subset \mathbb{R}^{m}$ are the $n_x$-dimensional state and the $m$-dimensional control input at time $k$, respectively. To synthesize control sequences later, let us also define $\mathbf{u}_{k:T-1}\coloneqq u_k \hdots u_{T-1}$ and $\mathbf{x}_{k+1:T}\coloneqq x_{k+1} \hdots x_{T}$.

The system in equation \eqref{eq:sys} operates among $N$ uncontrollable and stochastic dynamic agents, e.g., autonomous or remote-controlled robots or humans. We model the state of agent $i$ at time $k$ as a random variable $Y_{k,i}\in \mathcal{Y}_i \subset \mathbb{R}^{n_{y,i}}$ with state domain $\mathcal{Y}_i$, and we define the random state of all agents as $Y_k\coloneqq [Y_{k,1},\hdots,Y_{k,N}] \in \mathcal{Y}\subset \mathbb{R}^{n_y}$ with $n_y\coloneqq n_{y,1}+\hdots+n_{y,N}$ and $\mathcal{Y}\coloneqq \mathcal{Y}_1 \times \hdots \times \mathcal{Y}_N$. We do not assume to have any a-priori knowledge of the agent dynamics and intentions, and since the agents are uncontrollable, their finite-length trajectories are a-priori unknown. However, we assume that their  trajectories follow an unknown distribution $\mathcal{D}$, i.e., 
\[
    Y\coloneqq (Y_0, Y_1, \dots) \sim \mathcal{D}.
\]
We further assume that we have access to $\bar{K}$ realizations from $\mathcal{D}$.
\begin{assumption}\label{ass:data}
    We have access to $\bar{K}$ independent trajectories $Y^{(j)} \coloneqq (Y^{(j)}_1, Y^{(j)}_2, \dots)$ from the distribution $\mathcal{D}$. We split the trajectories into  calibration and training sets $D_{cal}\coloneqq \{Y^{(1)}, \dots, Y^{(K)}\}$ and $D_{train}\coloneqq \{Y^{(K+1)}, \dots, Y^{(\bar{K})}\}$, respectively.
\end{assumption}
Modeling dynamic agents by a distribution $\mathcal{D}$ includes general classes of systems such as discrete-time stochastic hybrid systems and Markov decision processes. We make the restriction that $\mathcal{D}$ does not depend on the behavior of the autonomous system in  \eqref{eq:sys}, i.e., $\mathcal{D}$ does not depend on $x$. Similar modeling assumptions have been made in \cite{lindemann2023safe} and are justified in many applications, e.g., when conservative control actions are taken in autonomous driving that do only alter the behavior of pedestrians minimally. By using robust or adaptive versions of conformal prediction, see e.g., \cite{cauchois2024robust,gibbs2021adaptive}, we can relax this assumption. Lastly, Assumption~\ref{ass:data} is not restrictive in most applications, e.g., past trajectory data of dynamic agents is available in robotic or transportation systems. 

In this paper, we consider the control of the system in  \eqref{eq:sys} under  temporal logic specifications that are jointly defined over the system and the  uncontrollable agents. Therefore, we define their joint state at time $k$ as $S_k \coloneqq [x_k,Y_{k}] \!\in\! \mathcal{X} \times \mathcal{Y} \subset \mathbb{R}^{n}$ with $n\coloneqq n_x+n_y$ and the joint state sequence as  $\mathbf{S}\coloneqq S_0,S_1,\hdots$.

\subsection{Signal Temporal Logic Control Synthesis among Uncontrollable Agents}\label{subsec:prob}
In this paper, we use signal temporal logic to describe the tasks defined over the system in \eqref{eq:sys} and uncontrollable agents. 
Signal temporal logic was originally introduced in \cite{maler2004monitoring}, and we here only give a brief introduction. 
The syntax of discrete-time finite-length STL formulae is
\begin{align}\label{eq:stl}
    \phi ::=   \top\mid \pi^\mu \mid \neg \phi \mid \phi_1 \wedge \phi_2 \mid \phi_1 \textbf{U}_{[a,b]} \phi_2,
\end{align}
where $\top$ is the \textsf{true} symbol ($\bot$ denotes the \textsf{false} symbol) and $\pi^\mu:\mathbb{R}^n \to \{\top, \bot \}$ is a predicate whose truth value   is determined by the sign of an underlying predicate function $\mu:\mathbb{R}^n \to \mathbb{R}$, i.e., for the state $s\in\mathbb{R}^n$ we have that $\pi^\mu(s)\coloneqq \top$ if and only if $\mu(s) \geq 0$.
The symbols $\neg$ and $\wedge$ denote the standard Boolean operators ``negation" and ``conjunction", respectively, and $\textbf{U}_{[a,b]}$ is the temporal operator ``\emph{until}" with $a \leq b, a,b\in \mathbb{N}$ with $\mathbb{N}$ being the natural numbers.
As we consider a discrete-time setting, we use $[a, b]$ to denote the discrete-time interval $[a, b] \cap \mathbb{N}$.
These operators can be used to define ``disjunction"  by $\phi_1 \vee \phi_2:=\neg(\neg \phi_1 \wedge \neg \phi_2)$, ``implication" by $\phi_1 \to \phi_2:= \neg \phi_1 \vee   \phi_2$, ``eventually" by $\mathbf{F}_{[a,b]} \phi:= \top \mathbf{U}_{[a,b]} \phi$, and ``always" by $\mathbf{G}_{[a,b]} \phi:=\neg \mathbf{F}_{[a,b]} \neg \phi$.

\textbf{STL Semantics.} STL formulae are evaluated over trajectories $\mathbf{s} \coloneqq s_0s_1\cdots$ where $\mathbf{s}$ could be a realization of $\mathbf{S}$.  
The notation $(\mathbf{s},k) \models \phi$ denotes that $\mathbf{s}$ satisfies the STL formula $\phi$ at time $k$. Particularly, we have 
$(\mathbf{s},k)\models \pi^\mu$ iff $\mu(s_k)\geq 0$,  $(\mathbf{s},k)\models \neg \phi$ iff $(\mathbf{s},k)\not\models \phi$, $(\mathbf{s},k)\models \phi_1 \wedge \phi_2$ iff $(\mathbf{s},k)\models \phi_1$ and $(\mathbf{s},k)\models \phi_2$,
and 
$(\mathbf{s},k)\models \phi_1 \mathbf{U}_{[a,b]} \phi_2$ 
iff $\exists k' \!\in\! [k+a, k+b]$ such that $(\mathbf{s},k') \models \phi_2$
and  $\forall k'' \!\in\! [k, k']$, we have $(\mathbf{s},k'') \models \phi_1$, i.e., 
$\phi_2$ will hold some time between $[k+a, k+b]$ in the future and until then $\phi_1$ always holds.  We write $\mathbf{s} \models \phi$ instead of $(\mathbf{s},0) \models \phi$ for simplicity.

\textbf{STL Quantitative Semantics.}
In addition to these \emph{qualitative} Boolean semantics, we can also define \emph{quantitative} semantics $\rho^\phi(\mathbf{s},k)\in\mathbb{R}$, also referred to as robust semantics, that describe how well $\phi$ is satisfied \cite{Alexandre2010}. Formally, we have 
\begin{align*}
    \rho^\top({\mathbf{s},k}) & \!\coloneqq\! \infty\\
    \rho^{\pi^\mu}({\mathbf{s},k})&\!\coloneqq\! \mu(s_k)\\
    \rho^{\neg \phi}({\mathbf{s},k})&\!\coloneqq\! -\rho^\phi({\mathbf{s},k})\\
    \rho^{\phi_1 \wedge \phi_2}({\mathbf{s},k})&\!\coloneqq\! \min(\rho^{\phi_1}({\mathbf{s},k}),\rho^{\phi_2}({\mathbf{s},k}))\\
    \rho^{\phi_1 \mathbf{U}_{[a,b]} \phi_2}({\mathbf{s},k})&\!\coloneqq\! \!\! \max_{k' \in[k+a, k+b]} \!\! \min(\rho^{\phi_2}({\mathbf{s},k'}), \!\! \min_{k'' \in [k, k']} \!\! \rho^{\phi_1}({\mathbf{s},k''})).
\end{align*}

\textbf{Problem Formulation.} In this paper, we consider bounded STL formulae $\phi$ where  time intervals $[a,b]$ in $\phi$ are bounded. The satisfaction of a bounded STL formula $\phi$ can be computed over bounded trajectories. Specifically, we require trajectories with a minimum length of $T_\phi$ where $T_\phi$ is the formula horizon formally defined in \cite{dokhanchi2014line, sadraddini2015robust}. For example, for $\phi \coloneqq \mathbf{F}_{[3,8]}\mathbf{G}_{[1,2]} \pi^{\mu} $, we have $T_\phi = 10$. We additionally assume that $\phi$ is in positive normal form, i.e., the formula $\phi$ contains no negations, and is build from convex and differentiable predicate functions $\mu$. The former assumption is without loss of generality as every STL formula can be re-written into an equivalent STL formula that is in positive normal form \cite{sadraddini2015robust,sadraddini2018formal}. The latter assumption is made for computational reasons and can be relaxed, since every STL formula can be transformed into positive normal form (see \cite{sadraddini2015robust} for more details). We summarize these assumptions next.

\begin{assumption}\label{ass:noneg}
We consider bounded STL formulae $\phi$ in positive normal form. We further assume that all  predicate functions $\mu$ in $\phi$ are convex and differentiable with respect to $Y_\tau$.
\end{assumption}

We impose specifications $\phi$ that satisfy Assumption \ref{ass:noneg} over the random trajectory $\mathbf{S}$, i.e., over the system trajectory $x_k$ and the random trajectories of uncontrollable dynamic agents $Y_k$. Formally, we consider the following problem in this paper.

\begin{myprob}\label{prob}
    Given the system in \eqref{eq:sys}, the uncontrollable agents $Y\coloneqq (Y_0, Y_1, \dots) \sim \mathcal{D}$, a set of  trajectories $D_{cal}$ and $D_{train}$ that satisfy Assumption \ref{alg:offline}, an STL formula $\phi$ from \eqref{eq:stl} that satisfies Assumption \ref{ass:noneg}, and a failure probability $\delta \!\in\! (0,1)$, compute the control inputs $u_k$ for all $k\in \{0,\hdots, T_\phi-1\}$ such that the STL formula is satisfied with a probability of at least $1-\delta$, i.e., 
    \[
        \text{Prob}(\mathbf{S} \models \phi) \geq 1-\delta.
    \]
\end{myprob}

The challenge for solving Problem \ref{prob} is the dependency of $\phi$ on the agents in $Y$, i.e., the joint trajectory $\mathbf{S}$ is only partially controllable. Therefore, our approach is to predict uncontrollable agents, construct prediction regions that capture uncertainty of these prediction, and synthesize the controller using this information. 
While our work is inspired by \cite{lindemann2023safe, dixit2023adaptive}, these works cannot be simply extended to solve our problem, since they consider state constraints of the form $c(x_k, y_k) \ge 0$ that are enforced at each time step separately. We instead consider STL specifications defined over the trajectory $\mathbf{S}$, where the task description cannot be separated into different time steps.
This temporal coupling makes the problem of designing closed-loop controllers different from existing works, since during the online execution we cannot discard constraints from previous time steps.
Furthermore, applying a Lipschitz reformulation to this problem as in \cite{lindemann2023safe, dixit2023adaptive} will bring conservative results.
Specifically, these works leverage the Lipschitz constant of the  function $c$ to capture the worst case over the prediction region.
In the STL case, taking the Lipschitz constant of the robustness function $\rho^\phi$ may be conservative and not practical.

Our solution consists of an offline and an online phase. In the offline phase, we train trajectory predictors using historical data and compute prediction regions via conformal prediction (Section 3). These regions capture the uncertainty of uncontrollable agents' trajectories with a probability of at least $1-\delta$. In the online phase, for open-loop control, we solve a mixed integer program (MIP) at the initial time step using the open-loop prediction regions, ensuring probabilistic satisfaction of the STL specification (Section 4). For closed-loop control, we iteratively update the predictions and prediction regions at each time step, solving the MIP at each time step and guaranteeing recursive feasibility and task satisfaction with high probability (Section 5). 

\subsection{Trajectory Predictors and Conformal Prediction}
\label{sec:CP}

\textbf{Trajectory Predictors. } From the training set $D_{train}$, we first train a trajectory predictor that estimates future states $(Y_{k+1},\hdots,Y_{T_\phi})$ from past observations $(Y_0,\hdots Y_k)$. Specifically, we denote the predictions made at time $k$ as $(\hat{Y}_{k+1|k}, \dots, \hat{Y}_{T_\phi|k})$. For instance, we can use recurrent neural networks (RNNs) \cite{salehinejad2017recent}, long short term memory (LSTM) networks \cite{yu2019review}, or  transformers~\cite{han2021transformer}. 
In this paper, we do not make assumption on the trajectory predictor.

\textbf{Conformal Prediction. } The predictions $\hat{Y}_{\tau|k}$ of $Y_{\tau}$ for $\tau> k$, which may be obtained from learning-enabled components, may not always be accurate. We will use conformal prediction to obtain prediction regions for $Y_{\tau}$ that are valid with high probability. We leverage conformal prediction which is a statistical tool for uncertainty quantification with minimal assumptions \cite{shafer2008tutorial,vovk2005algorithmic}.  

Let $R^{(0)}, \dots, R^{(K)} \in \mathbb{R}$ be $K+1$ i.i.d. random variables which are referred to as the \emph{nonconformity score}.
For instance, we may define the prediction error $R^{(j)} \coloneqq ||Y_\tau^{(j)} - \hat{Y}_{\tau|k}^{(j)}||$ at time $\tau>k$ for all calibration trajectories $Y^{(j)}\in D_{cal}$. Naturally, a small value of $R^{(j)}$ in this case implies accurate predictions. We later present  nonconformity scores to obtain prediction regions for uncontrollable agents over multiple time steps. Given a failure probability $\delta \!\in\! (0,1)$, our goal is to compute  $C$ from $R^{(1)},\dots, R^{(K)}$ such that the probability of $R^{(0)}$ being bounded by $C$ is not smaller than $1-\delta$, i.e., such that
\[
    \text{Prob}(R^{(0)} \leq C) \geq 1-\delta.
\]
By the results in Lemma 1 of \cite{tibshirani2019conformal}, we can compute 
\begin{align*}
    C:=\text{Quantile}_{1-\delta}(R^{(1)}, \dots, R^{(K)},\infty)
\end{align*}
as the $(1-\delta)$th quantile over the empirical distribution of  $R^{(1)}, \dots, R^{(K)}$  and $\infty$.\footnote{Formally, we define the quantile function as $\text{Quantile}_{1-\delta}(R^{(1)}, \dots, R^{(K)},\infty):=\text{inf}\{z\in \mathbb{R}|\text{Prob}(Z\le z)\ge 1-\delta\}$ with the random variable $Z:=1/(K+1)(\sum_i \delta_{R^{(i)}}+\delta_{\infty})$ where $\delta_{R^{(i)}}$ and $\delta_{\infty}$ are dirac distributions centered at $R^{(i)}$ and $\infty$, respectively.}
By assuming that $R^{(1)}, \dots, R^{(K)}$ are sorted in non-decreasing order and by adding $R^{(K+1)} \coloneqq \infty$, we have that $C = R^{(p)}$ where $p \coloneqq \lceil (K+1)(1-\delta) \rceil$, i.e., $C$ is the $p$th smallest nonconformity score. To obtain meaningful prediction regions, we remark that $K\ge \lceil (K+1)(1-{\delta})\rceil$ has to hold. Otherwise we get $C=\infty$. Lastly, note that the guarantees $\text{Prob}(R^{(0)} \leq C) \geq 1-\delta$ are marginal over the randomness in $R^{(0)}, R^{(1)}, \hdots, R^{(K)}$ (with $C$ depending on $ R^{(1)}, \hdots, R^{(K)}$) as opposed to being conditional on $ R^{(1)}, \hdots, R^{(K)}$.

In the following, the open-loop controller is introduced in Sections \ref{sec:cp} and \ref{sec:solution_open}, and we propose the closed-loop controller in Section \ref{sec:solution_closed}.

\section{Open-Loop Prediction Regions for Multiple Uncontrollable Agents}\label{sec:cp}

To solve Problem \ref{prob}, we  use trajectory predictions and compute probabilistic prediction regions for each uncontrollable agent that are valid over multiple time steps. 
Recall that $Y_{\tau,i}$ denotes the random state of agent $i$ at time $\tau$, and let $\hat{Y}_{\tau|0,i}$ denote the prediction of  $Y_{\tau,i}$ made at time $k=0$. 

We are interested in constructing prediction regions for all predictions $\tau \!\in\! \{1, \dots, T_{\phi}\}$ made at time $k\coloneqq 0$  and  for all agents $i \!\in\!  \{1, \dots, N\}$, i.e., such that 
\begin{align}\label{eq:cp_open_}
    \text{Prob}(||Y_{\tau,i} - \hat{Y}_{\tau|0,i}|| \leq C_{\tau|0,i}, \forall (\tau,i) \!\in\! \{1, \dots, T_{\phi}\}\times \{1, \dots, N\}) \!\geq\! 1-\delta, 
\end{align}
where  $C_{\tau|0,i}$ indicates the $\tau$-step ahead prediction error of agent $i$ for predictions $\hat{Y}_{\tau|0,i}$.
In the following, we show how to compute these prediction regions using conformal prediction, and the formal result is given in Theorem \ref{thm:proba}.

\textbf{Data-efficient and accurate prediction regions.} Computing $\tau$-step ahead prediction regions via $C_{\tau|0,i}$ is in general difficult. Existing works \cite{lindemann2023safe,stankeviciute2021conformal} are conservative, data-inefficient, and do not provide prediction regions for multiple agents. The reason is that $C_{\tau|0, i}$ are independently computed for all $(\tau,i) \!\in\! \{1, \dots, T_{\phi}\}\times \{1, \dots, N\}$. This independent computation is followed by  a conservative and inefficient union bounding argument (more details below). 

Instead, we obtain these prediction regions jointly by extending the framework in \cite{zhao2024robust} to incorporate multiple agents. For the open-loop case, we define the normalized nonconformity score
\begin{align}
    & R_{OL}^{(j)} \coloneqq  \max_{\substack{(\tau,i) \in \{1,\dots, T_\phi\} \\ \times \{1,\dots, N\}}}
    \frac{||Y_{\tau,i}^{(j)}-\hat{Y}_{\tau|0,i}^{(j)}||}{\sigma_{\tau|0, i}}\label{eq:nonscore_open}  
\end{align}
for calibration trajectories $Y^{(j)}\in D_{cal}$  where $\sigma_{\tau|0, i}>0$ are constants that normalize the prediction error to  $[0,1]$ for all times $\tau$ and agents $i$. This nonconformity score is inspired by \cite{cleaveland2024conformal}, where a mixed integer linear complementarity program is used to find optimal constants $\sigma_{\tau|0, i}$.
However, we do not need to solve nonconvex optimization problems  and instead simply normalize over the training dataset $D_{train}$ which has practical advantages. Specifically,  we use the normalization $\sigma_{\tau|0, i}:=\max_{j} \|Y^{(j)}_{\tau,i} - \hat{Y}^{(j)}_{\tau|0,i}\|$ for training trajectories  $Y^{(j)}\in D_{train}$ with the assumption that $\sigma_{\tau|0, i}$ is non-zero. 
This is important as the optimization problem in \cite{cleaveland2024conformal} may have no solution, or the optimal solution cannot be found due to the nonconvexity of the problem, or it takes too long to solve in practice, especially for large $T_\phi$ and $N$. The reader can find an empirical comparison in Section \ref{subsec:case1}, where we show that we can obtain similar prediction regions as \cite{cleaveland2024conformal}, but in a computationally more efficient manner.

The intuitions behind the nonconformity score in \eqref{eq:nonscore_open} are that normalization via $\sigma_{\tau|0, i}$ will prevent the prediction error $||Y_{\tau,i}^{(j)}-\hat{Y}_{\tau|0,i}^{(j)}||$ for a specific agent and time to dominate the $\max$ operator. This would result in overly conservative prediction regions for other times and agents.

\textbf{Valid Prediction Regions with Conformal Prediction. } We can now apply conformal prediction, as introduced in Section \ref{sec:CP}, to the open-loop nonconformity score  $R^{(j)}_{OL}$. By re-normalizing these nonconformity scores, we then obtain prediction regions as in equations  \eqref{eq:cp_open_}.

\begin{mythm}\label{thm:proba}
    Given the random trajectory $Y\coloneqq (Y_0, Y_1, \dots) \sim \mathcal{D}$, a set of trajectories $D_{cal}$ and $D_{train}$ that satisfy Assumption \ref{ass:data}, and the failure probability $\delta \!\in\! (0,1)$, then the prediction regions in equation \eqref{eq:cp_open_} hold for the choice of 
    \begin{align}
        C_{\tau|0,i} & \coloneqq C_{OL} \sigma_{\tau|0, i}, \label{eq:C_open}
    \end{align}
    where $\sigma_{\tau|0, i}$ are positive constants and
    \begin{align*}
        C_{OL}&:=\text{Quantile}_{1-\delta}(R_{OL}^{(1)},\hdots,R_{OL}^{(K)}, \infty).
    \end{align*} 
\end{mythm}

\proof{
For a test trajectory $Y\coloneqq (Y_0, Y_1, \dots) \sim \mathcal{D}$, we immediately know from conformal prediction that 
\begin{align}\label{eq:open_proof}
    \text{Prob}\Big(\max_{\substack{(\tau,i)\in\{1,\dots,T_{\phi}\} \\ \times\{1,\dots,N\}}}
    \frac{||Y_{\tau,i}\!-\!\hat{Y}_{\tau|0,i}||}{\sigma_{\tau|0, i}} \leq C_{OL}\Big) \geq 1\!-\!\delta
\end{align}
by construction of $C_{OL}$. We then see that equation \eqref{eq:open_proof} implies that 
\begin{align}\label{eq:open_proof_}
    \text{Prob}\Big(&\frac{||Y_{\tau,i}-\hat{Y}_{\tau|0,i}||}{\sigma_{\tau|0, i}} \leq C_{OL}, \forall (\tau,i) \!\in\!\{1,\dots, T_\phi\} \times \{1,\dots, N\}\Big) \geq 1-\delta.
\end{align}
Since $\sigma_{\tau|0, i}>0$, we can multiply the inequality in \eqref{eq:open_proof_} with $\sigma_{\tau|0, i}$. We then see that \eqref{eq:cp_open_} is valid with the choice of $C_{\tau|0,i} \coloneqq C_{OL} \sigma_{\tau|0, i}$ which completes the proof.
}

\begin{remark}
    In this work, the prediction regions $C_{\tau|0,i}$ are norm balls because the nonconformity score is defined using the Euclidean norm. This choice was made for simplicity and interpretability. However, the conformal prediction framework is flexible and can accommodate other nonconformity scores, which would lead to different prediction regions, e.g., ellipsoids \cite{messoudi2022ellipsoidal} or multi-modal regions \cite{tumu2024multi}. While we chose norm balls for simplicity, the proposed framework is general and can support any convex or even nonconvex shape by defining the nonconformity score appropriately.
\end{remark}

\textbf{Comparison with existing work.}   With slight modification of \cite{lindemann2023safe,stankeviciute2021conformal}, an alternative approach for obtaining $C_{\tau|0,i}$  would be to compute $C_{\tau|0,i}\coloneqq \text{Quantile}_{1-\delta/(T_\phi N)}(R^{(1)}, \dots, R^{(K)},\infty)$ with  the nonconformity score $R^{(j)}\coloneqq ||Y_{\tau,i} - \hat{Y}_{\tau|0,i}||$. 
However, it is evident that taking the $1-\delta$ quantile in Theorem \ref{thm:proba} compared to the $1-\delta/(T_\phi N)$ quantile is much more data inefficient. Specifically, for a fixed $\delta\in (0,1)$, we only require  $K\ge (1-\delta)/\delta$ calibration trajectories as opposed to $K\ge (T_\phi N-\delta)/\delta$ calibration trajectories.\footnote{Computing the $1-\delta$ and $1-\delta/(T_\phi N)$ quantiles requires that $\lceil (K+1)(1-\delta) \rceil \le K$ and $\lceil (K+1)(1-\delta/(T_\phi N)) \rceil\le K$, respectively. } Intuitively, it is also evident that our choice of normalization constants $\sigma_{\tau|0,i}$ provides less conservative prediction regions in practice as evaluating the $1-\delta$ quantile is more favorable compared to the $1-\delta/(T_\phi N)$ for larger task horizons $T_\phi$ and number of agents $N$. Finally, we remark that one could even consider obtaining conditional guarantees using conditional conformal prediction \cite{gibbs2023conformal}. Similarly, we can use versions of robust conformal prediction to be robust against distribution shifts  \cite{cauchois2024robust,tibshirani2019conformal, zhao2024robust}, e.g., caused by interaction between the system in \eqref{eq:sys} and the distribution $\mathcal{D}$.

\textbf{For practitioners. } We summarize our method to obtain the prediction regions in equation \eqref{eq:cp_open_} via Theorem \ref{thm:proba} in Algorithm \ref{alg:offline}, which is executed offline. We note that, in order to compute $C_{OL}$, we sort the corresponding nonconformity scores (after adding $R_{OL}^{(K+1)}:=\infty$) in nondecreasing order and set $C_{OL}$ to be the $p$th smallest nonconformity score where $p \coloneqq \lceil (K+1)(1-\delta) \rceil$, as explained in Section \ref{sec:CP}. 
Specifically, the variable $p$ is set in line 1. 
We then compute the predictions $\hat{Y}_{\tau|0,i}^{(j)}$ for all calibration data in lines 2-5 and calculate the normalization constants $\sigma_{\tau|0,i}$ over the training data in line 6. 
Finally, we compute the nonconformity scores $R_{OL}^{(j)}$ for all calibration data in lines 7-9 and apply conformal prediction to obtain $C_{OL}$ in lines 10-11. According to Theorem \ref{thm:proba}, we then obtain the open-loop prediction regions $C_{\tau|0,i}$.

\begin{algorithm}[ht]
	\caption{Multi-Agent Prediction Regions (Offline)}
	\label{alg:offline}
	\KwIn{failure probability $\delta$, dataset $D$, horizon $T_\phi$}
	\KwOut{Constants $C_{OL}$ and $\sigma_{\tau|0,i}$} 

    $p \gets \lceil (K+1)(1-\delta) \rceil$ \\
	\ForAll{$i \in \{1,\dots, N\}$}
	{	
        \ForAll{$\tau \in [k+1, T_\phi]$}
        {
            \ForAll{$j\in \{1,\hdots,\bar{K}\}$}
            {
                Compute  $\hat{Y}_{\tau|0, i}^{(j)}$ from  $(Y_{0,i}^{(j)}, \dots, Y_{k,i}^{(j)})$ 
            }
            $\sigma_{\tau|0,i} \gets \max_{j \in \{K+1,\dots,\bar{K}\}}||Y_{\tau,i}^{(j)}-\hat{Y}_{\tau|0,i}^{(j)}||$
        }
	}
    \ForAll{$j \in \{1,\hdots,K\}$}
    {
        Compute $R_{OL}^{(j)}$ as in equation \eqref{eq:nonscore_open}
    }
    Set $R^{(K+1)}_{OL} \gets \infty$\\
    Sort $R^{(j)}_{OL}$ in nondecreasing order \\
    Set $C_{OL} \gets R^{(p)}_{OL}$
\end{algorithm}



\section{Open-Loop Predictive STL Control Synthesis}\label{sec:solution_open}
In this section, we use the previously computed prediction regions and design predictive controllers that solve Problem \ref{prob}, i.e., controllers $u$  that ensure that $x$ is such that $\text{Prob}(\mathbf{S} \models \phi) \geq 1-\delta$. Specifically, we use the prediction regions and present an MIP encoding for the STL task $\phi$. We first present a qualitative and a quantitative encoding using the Boolean and quantitative semantics of $\phi$ in Sections \ref{subsec:qua} and \ref{subsec:quan}, respectively. In Section \ref{subsec:loop}, we then use these encodings and present the open-loop controller.

\subsection{Qualitative STL Encoding with Multi-Agent Prediction Regions}\label{subsec:qua}

A standard way of encoding STL tasks for deterministic trajectories $\mathbf{s}$ is by using an MIP encoding \cite{raman2014model,raman2015reactive}. The idea is to introduce a binary variable $z_0^\phi$ and a set of mixed integer constraints over $\mathbf{s}$ such that $z_0^\phi = 1$ if and only if $\mathbf{s}\models \phi$. In this paper, however, we deal with stochastic trajectories $\mathbf{S}$ that consist of the system trajectory $x$ and the stochastic trajectories of dynamic agents $Y$. 
We will instead introduce a binary variable $\bar{z}_0^\phi$ and a set of mixed integer constraints over $x$ and the prediction regions in \eqref{eq:cp_open_} such that $\bar{z}_0^\phi = 1$ implies that $\mathbf{S}\models \phi$ with a probability of at least $1-\delta$. In this way, $\bar{z}_0^\phi$ can be seen as a probabilistic version of $z_0^\phi$. We emphasize that our MIP encoding provides sufficient but not necessary conditions for the task satisfaction. Specifically, we can not guarantee necessity due to the use of probabilistic prediction regions.

In the remainder, we generate the mixed integer constraints that define $\bar{z}_0^\phi$ recursively on the structure of $\phi$. Our main innovation is a probabilistic MIP encoding for predicates, while the encoding of Boolean and temporal operators follows \cite{raman2014model,raman2015reactive}. We here recall from Assumption \ref{ass:noneg} that the formula $\phi$ is in positive normal form so that we do not need to consider negations. To provide a general encoding that can be used for open-loop and closed-loop control, we let $k\ge 0$ denote the current time.  
For the open-loop case, we only need to consider the initial time, i.e., $k=0$.
As we will treat the system state $x_k$ separately from the state of uncontrollable dynamic agents $Y_k$ contained in $S_k$, we will write  $\mu(x_k, Y_k)$ instead of $\mu(s_k)$.

\textbf{Predicates $\pi^\mu$.} Let us denote the set of predicates $\pi^\mu$ in $\phi$ by $\mathcal{P}$. Now, for each predicate $\pi^\mu\in\mathcal{P}$ and for each time $\tau\ge 0$, we introduce a binary variable $\bar{z}_{\tau|k}^{\mu} \!\in\! \{0,1\}$.
If $\tau \leq k$, then we have observed the value of $Y_\tau$ already, and we set $\bar{z}_{\tau|k}^{\mu} = 1$ if and only if $\mu(x_\tau, Y_\tau) \geq 0$. If $\tau > k$, then we would like to constrain $\bar{z}_{\tau|k}^{\mu} $ such that  $\bar{z}_{\tau|k}^{\mu} = 1$ implies $\text{Prob}(\mu(x_{\tau|k}, Y_{\tau}) \geq 0) \geq 1 - \delta$. Motivated by Theorem \ref{thm:proba}, we use the Big-M method and define the constraint
\begin{align}\label{eq:predicate_}
    - \min_{y \in \mathcal{B}_{\tau|k}} \mu(x_{\tau|k}, y) & \leq M (1-\bar{z}_{\tau|k}^{\mu}) -\epsilon, 
\end{align}
where $M$ and $\epsilon$ are sufficiently large and sufficiently small positive constants, respectively, see \cite{raman2014model, bemporad1999control} for more details. The minimization of $\mu(x_{\tau|k}, y)$ over the $y$ component within the set $\mathcal{B}_{\tau|k}$ will account for all $Y_{\tau}$ that are contained within our probabilistic prediction regions. Specifically, the set $\mathcal{B}_{\tau|k}$ contains all states contained within a geometric norm ball that is centered around the predictions $\hat{Y}_{\tau|k,i}$ with a size of $C_{\tau|k, i}$ and is defined as
\begin{align}\label{eq:B_tau}
    \mathcal{B}_{\tau|k} \coloneqq  \{ [y_{1}, \dots, y_{N}]^\top \!\in\! \mathbb{R}^{n_y} \mid \forall i\in \{1,\hdots,N\}, ||y_{i} -\hat{Y}_{\tau|k, i}|| \leq C_{\tau|k, i}\}. 
\end{align}
In the case when $k=0$, we use the values of $C_{\tau|0, i}$ as in \eqref{eq:C_open} that define the open-loop prediction region in \eqref{eq:cp_open_}. This implies that $\text{Prob}(Y_\tau\in \mathcal{B}_\tau, \tau\in\{1,\hdots,T_\phi\})\ge 1-\delta$.

By Theorem \ref{thm:proba} and the construction in equation \eqref{eq:predicate_}, we have the following straightforward result.

\begin{mycol}\label{col:predicate}
    If $\bar{z}_{\tau|0}^\mu = 1$ for all pairs $(\tau, \pi^\mu) \in TP$, where $TP \subseteq \{0, \dots, T_\phi\}\times \mathcal{P}$ is the set of some pairs, then it holds that $\text{Prob}(\mu(x_{\tau|0}, Y_{\tau}) \geq 0, \forall (\tau, \pi^\mu)\in TP)\ge 1-\delta$. 
\end{mycol}
\proof{
    If $\bar{z}_{\tau|0}^\mu = 1$ for all pairs $(\tau, \pi^\mu) \in TP$, then it follows by equation \eqref{eq:predicate_} that $x_{\tau|0}$ is such that $0< \epsilon\le \text{min}_{y \in \mathcal{B}_{\tau|0}} \mu(x_{\tau|0}, y)$ for all pairs $(\tau, \pi^\mu) \in TP$. Now, using Theorem \ref{thm:proba} (Equation \eqref{eq:cp_open_}), we can directly conclude that $\text{Prob}(\mu(x_{\tau|0}, Y_{\tau}) \geq 0, \forall (\tau, \pi^\mu)\in TP)\ge 1-\delta$.
}

\begin{remark}
    We do not need to enforce $\bar{z}_{\tau|k}^\mu = 0$ since the STL formula $\phi$ is in positive normal form as per Assumption \ref{ass:noneg}. However, by using the constraint $\text{max}_{y \in \mathcal{B}_{\tau|k}} \mu(x_{\tau|k}, y) \leq M \bar{z}_{\tau|k}^{\mu} - \epsilon$ we could enforce that $\bar{z}_{\tau|k}^\mu = 0$ implies  $\text{Prob}(\mu(x_{\tau|k}, Y_{\tau}) < 0, \forall (\tau, \pi^\mu)\in TP)\ge 1-\delta$.
\end{remark}

\textbf{Computational considerations. } Note that we minimize $\mu(x_{\tau|k},y)$ over the $y$ component within the ball $\mathcal{B}_{\tau|k}$  in equation \eqref{eq:predicate_}. In an outer loop, we will additionally need to  optimize over $x_{\tau|k}$ to compute control inputs $u_{\tau|k}$. To obtain computational tractability for this, we will now remove the minimum over $y$ and instead write equation \eqref{eq:predicate_} in closed-form. Since the function $\mu$ is continuously differentiable and convex in its parameters by Assumption \ref{ass:noneg}, we can use the KKT conditions of the inner problem to do so. Specifically, the constraints in equation \eqref{eq:predicate_} can be converted into a bilevel optimization problem where the outer problem consists of the constraints
\begin{align}\label{eq:kkt_}
\begin{split}
    - \mu(x_{\tau|k}, y^*) & \leq M (1-\bar{z}_{\tau|k}^{\mu}) -\epsilon, 
    \end{split}
\end{align}
and where the inner optimization problem is 
\begin{subequations}\label{eq:cons_y}
    \begin{align}
        & y^*:=\underset{y}{\text{argmin}} \;\mu(x_{\tau|k}, y), \\
        & \text{s.t.} \; 
        ||y_{i} -\hat{Y}_{\tau|k, i}|| \leq C_{\tau|k, i}, \forall i \!\in\! \{1, \dots, N\}. 
    \end{align}
\end{subequations}
We can use the KKT conditions of the inner optimization problem in \eqref{eq:cons_y} after a minor modification in which we change the constraints $||y_i -\hat{Y}_{\tau|k, i}|| \leq C_{\tau|k, i}$ into $(y_i -\hat{Y}_{\tau|k, i})^\top (y_i-\hat{Y}_{\tau|k,i}) \leq C_{\tau|k, i}^2$ to obtain differentiable constraint functions. Formally, we obtain
\begin{subequations}\label{eq:kkt}
    \begin{align}
        & \frac{\partial \mu(x_{\tau|k}, y^*)}{\partial y} + \Sigma_{i=1}^{N} \lambda_i \frac{\partial  (y_i^*-\hat{Y}_{\tau|k,i})^\top (y_i^*-\hat{Y}_{\tau|k,i})}{\partial y_i} = 0  \\
        & (y_i^*-\hat{Y}_{\tau|k,i})^\top (y_i^*-\hat{Y}_{\tau|k,i}) \leq C_{\tau|k, i}^2, \forall i \!\in\! \{1, \dots, N\}  \\
        & \lambda_i \geq 0, \forall i \!\in\! \{1, \dots, N\}  \\
        & \lambda_i ((y_i^*-\hat{Y}_{\tau|k,i})^\top (y_i^*-\hat{Y}_{\tau|k,i}) - C_{\tau|k, i}^2) = 0, \forall i \!\in\! \{1, \dots, N\} 
    \end{align}
\end{subequations}
as the set of KKT conditions for the inner problem in \eqref{eq:cons_y}.

We have that a solutions $y^*$ to \eqref{eq:cons_y} is a feasible solution to the KKT conditions in \eqref{eq:kkt} and vice versa \cite{boyd2004convex}, since Slater's condition holds for the inner convex optimization problem \eqref{eq:cons_y}, with the reasoning that $C_{\tau|k, i} >0$.\footnote{Here, we make the mild assumption that $C_{OL} \neq 0$. Note that $C_{OL}$ is computed from the nonconformity score $R^{(j)}$ in \eqref{eq:nonscore_open} which is only zero if our trajectory predictor is perfect, i.e., if $||Y_{\tau,i}^{(j)}-\hat{Y}_{\tau|0,i}^{(j)}||=0$, which is rarely the case. We could  modify the nonconformity score $R^{(j)}$ by adding a small constant to be guaranteed to achieve $C_{OL} \neq 0$  at the expense of making the prediction region slightly more conservative. }
As a result, we can replace \eqref{eq:predicate_} with equations \eqref{eq:kkt_} and \eqref{eq:kkt} equivalently.

\textbf{Boolean and temporal operators.}
So far, we have  encoded predicates $\pi^\mu$. We can encode Boolean and temporal operators recursively using the standard MIP encoding \cite{raman2014model,raman2015reactive}. In brief, for the conjunction $\phi:=\wedge_{i=1}^m \phi_i$ we introduce the binary variable $\bar{z}_{\tau|k}^{\phi} \!\in\! \{0,1\}$ that is such that $\bar{z}_{\tau|k}^{\phi} =1$ if and only if  $\bar{z}_{\tau|k}^{\phi_1}=\hdots=\bar{z}_{\tau|k}^{\phi_m} =1$. Consequently, $\bar{z}_{\tau|k}^{\phi} =1$ implies that $(\mathbf{S},\tau)\models \phi$ holds with a probability of at least $1-\delta$ at time step $k$. We achieve this by enforcing the constraints
    \begin{align}
        & \bar{z}_{\tau|k}^{\phi} \leq \bar{z}_{\tau|k}^{\phi_i}, i=1,\dots, m, \nonumber \\
        & \bar{z}_{\tau|k}^{\phi} \geq 1-m+\Sigma_{i=1}^m \bar{z}_{\tau|k}^{\phi_i}. \nonumber
    \end{align}
For the disjunction $\phi:=\vee_{i=1}^m \phi_i$, we follow the same idea but instead enforce the constraints
\begin{align}
    & \bar{z}_{\tau|k}^{\phi} \geq \bar{z}_{\tau|k}^{\phi_i}, i=1,\dots, m, \nonumber \\
    & \bar{z}_{\tau|k}^{\phi} \leq \Sigma_{i=1}^m \bar{z}_{\tau|k}^{\phi_i}. \nonumber
\end{align}


For temporal operators, we follow a similar procedure, but first note that temporal operators can be written as Boolean combinations of conjunctions and disjunctions. For $\phi:=\mathbf{G}_{[a,b]}\phi_1$, we again introduce a binary variable $\bar{z}_{\tau|k}^{\phi} \!\in\! \{0,1\}$ and encode $\bar{z}_{\tau|k}^{\phi} = \bigwedge_{\tau'=\tau+a}^{\tau+b} \bar{z}_{\tau'|k}^{\phi_1}$ using the MIP constraints for conjunctions introduced before. Similarly, for $\phi:=\mathbf{F}_{[a,b]}\phi_1$ we use that $\bar{z}_{\tau|k}^{\phi} = \bigvee_{\tau'=\tau+a}^{\tau+b} \bar{z}_{\tau'|k}^{\phi_1}$, while for $\phi=\phi_1\mathbf{U}_{[a,b]}\phi_2$ we use that $\bar{z}_{\tau|k}^{\phi} = \bigvee_{\tau'=\tau+a}^{\tau+b} (\bar{z}_{\tau'|k}^{\phi_2} \wedge \bigwedge_{\tau''=\tau}^{\tau'} \bar{z}_{\tau''|k}^{\phi_1})$.


\textbf{Soundness of the encoding.} At the initial time step $k=0$, the procedure described above provides us with a binary variable $\bar{z}_{0|k}^\phi$ and a set of mixed integer constraints over $x$ such that $\bar{z}_{0|k}^\phi = 1$ implies that $\mathbf{S}\models \phi$ holds with a probability of at least $1-\delta$. We summarize this result next, and we introduce of the result of closed-loop in the next section.

\begin{mythm}[Open-loop qualitative encoding result]\label{theorem_qual}
    Let the conditions from Problem \ref{prob} hold.  If $\bar{z}_{0|0}^\phi = 1$, then we have that $\text{Prob}(\mathbf{S}\models \phi)\ge 1-\delta$.

\end{mythm}

\proof{By enforcing that $\bar{z}_{0|0}^\phi = 1$, there is an assignment of $0$ and $1$ values to the binary variables $\bar{z}_{\tau|0}^\mu$ for each pair $(\tau, \pi^\mu)\in\{1,\hdots,T_\phi\} \times \mathcal{P}$.  
This assignment is non-unique. 
Specifically, let $\bar{z}_{\tau|0}^\mu = 1$ for all pairs $(\tau, \pi^\mu) \in TP$, where $TP \subseteq \{1, \dots, T_\phi\}\times \mathcal{P}$ is the set of some pairs, and let $\bar{z}_{\tau|0}^\mu = 0$ otherwise. 
By Corollary \ref{col:predicate}, it follows that $\text{Prob}(\mu(x_{\tau|0}, Y_{\tau}) \geq 0, \forall (\tau , \pi^\mu)\in TP)\ge 1-\delta$. By the recursive definition of $\bar{z}_{0|0}^\phi$ using the encoding from \cite{raman2014model,raman2015reactive} for Boolean and temporal operators and since $\phi$ is in positive normal form,  it follows that $\text{Prob}(\mathbf{S}\models \phi)\ge 1-\delta$.}

\subsection{Quantitative STL Encoding with Multi-Agent Prediction Regions}\label{subsec:quan}

The qualitative MIP encoding ensures that $\bar{z}_{0|0}^\phi = 1$ implies that $\text{Prob}(\mathbf{S}\models \phi)\ge 1-\delta$ as in Theorem \ref{theorem_qual}. However, in some cases one may want to  optimize over the quantitative semantics $\rho^\phi(\mathbf{S},0)$. We next present a quantitative MIP encoding for $\phi$. Specifically, we will recursively define a continuous variable $\bar{r}^\phi_{0|k}\in \mathbb{R}$ that will be such that $\rho^\phi(\mathbf{S},0) \geq \bar{r}^\phi_{0|k}$ with a probability of at least $1 - \delta$. Our main innovation is a quantitative MIP encoding for predicates, while the Boolean and temporal operators again follow the standard MIP encoding \cite{raman2014model,raman2015reactive}.



\textbf{Predicates. } For each predicate $\pi^\mu\in\mathcal{P}$ and for each time $\tau\ge 0$, we introduce a continuous variable $\bar{r}_{\tau|k}^{\mu}\in\mathbb{R}$. Inspired by \cite{lindemann2023conformal}, we define $\bar{r}_{\tau|k}^{\mu}$ as
\begin{align}\label{eq:rho_bar}
    \bar{r}_{\tau|k}^{\mu} \coloneqq 
    \left\{
        \begin{array}{ll}
        \mu(x_\tau, Y_\tau) \ &\text{if} \ \tau \leq k, \\
        \min_{y \in \mathcal{B}_{\tau|k}} \mu(x_{\tau|k}, y) & \text{otherwise}
        \end{array}
    \right. 
\end{align} 
where we again notice that $Y_{\tau}$ is known if $\tau \leq k$, while we compute $\min_{y \in \mathcal{B}_{\tau|k}} \mu(x_{\tau|k}, y)$ if $\tau > k$ to consider the worst case of $\mu(x_{\tau|k}, y)$ for $y$ within the prediction region $\mathcal{B}_{\tau|k}$.

By Theorem \ref{thm:proba}, which ensures that $Y_{\tau} \in \mathcal{B}_{\tau|k}$ holds with a probability of at least $1-{\delta}$, and the construction in equation \eqref{eq:rho_bar}, we have the following straightforward result.

\begin{mycol}\label{col:predicate_quan}
    It holds that $\text{Prob}(\rho^\mu(\mathbf{S},\tau) \geq \bar{r}^\mu_{\tau|0}, \forall (\tau, \pi^\mu)\in\{0, \dots, T_\phi\}\times \mathcal{P})\ge 1-\delta$, where $\mathbf{S}_\tau = [x_{\tau|0}, Y_{\tau}]$.
\end{mycol}
\proof{
The result follows directly from Theorem \ref{thm:proba} (Equation \eqref{eq:cp_open_}) and the construction of $\bar{r}^\mu_{\tau|0}$ in equation \eqref{eq:rho_bar} using the prediction region $\mathcal{B}_{\tau|0}$.
}

Similarly to the qualitative encoding, the term $\min_{y \in \mathcal{B}_{\tau|k}} \mu(x_{\tau|k}, y)$ in equation \eqref{eq:rho_bar} can be written as in equation \eqref{eq:cons_y}, and we can use equations \eqref{eq:kkt} instead of \eqref{eq:cons_y}.

\textbf{Boolean and temporal operators. } As for the qualitative encoding, we only provide a brief summary for the quantitative encoding of Boolean and temporal operators that follow standard encoding rules.
For the conjunction $\phi:=\wedge_{i=1}^m \phi_i$, we introduce the continuous variable $\bar{r}_{\tau|k}^{\phi}\in\mathbb{R}$ that will be such that $\bar{r}_{\tau|k}^{\phi} = \min \{\bar{r}_{\tau|k}^{\phi_1}, \hdots, \bar{r}_{\tau|k}^{\phi_m}\}$. Particularly, we achieve this by enforcing the constraints
\begin{align}
    & \Sigma_{i=1}^m p_{\tau|k}^{\phi_i} = 1, \nonumber \\
    & \bar{r}_{\tau|k}^{\phi} \leq \bar{r}_{\tau|k}^{\phi_i}, i=1,\dots, m, \nonumber \\
    & \bar{r}_{\tau|k}^{\phi_i} - (1-p_{\tau|k}^{\phi_i}) M \leq \bar{r}_{\tau|k}^{\phi} \leq \bar{r}_{\tau|k}^{\phi_i} + M(1-p_{\tau|k}^{\phi_i}), i=1,\dots, m, \nonumber
\end{align}
where $p_{\tau|k}^{\phi_i}\in\{0,1\}$ are $m$ new binary variables so that $\bar{r}_{\tau|k}^{\phi} = r_{\tau|k}^{\phi_i}$ if and only if $p_{\tau|k}^{\phi_i} = 1$. By this encoding, $\bar{r}_{\tau|k}^{\phi} \geq 0$ implies that $(\mathbf{S},\tau)\models \phi$ holds with a probability of at least $1-\delta$.  For the disjunction $\phi:=\vee_{i=1}^m \phi_i$, we follow the same idea but instead enforce the constraints
\begin{align}
    & \Sigma_{i=1}^m p_{\tau|k}^{\phi_i} = 1, \nonumber \\
    & \bar{r}_{\tau|k}^{\phi} \geq \bar{r}_{\tau|k}^{\phi_i}, i=1,\dots, m, \nonumber \\
    & \bar{r}_{\tau|k}^{\phi_i} - (1-p_{\tau|k}^{\phi_i}) M \leq \bar{r}_{\tau|k}^{\phi} \leq \bar{r}_{\tau|k}^{\phi_i} + M(1-p_{\tau|k}^{\phi_i}), i=1,\dots, m. \nonumber
\end{align}
For temporal operators, we again note that we can write each operator as Boolean combinations of conjunctions and disjunctions and then follow the same procedure as for the qualitative encoding.

\textbf{Soundness of the encoding. } At the initial time step $k=0$, the procedure described above gives us  a continuous variable $\bar{r}_{0|k}^\phi$ and a set of mixed integer constraints over $x$ such that $\bar{r}_{0|k}^\phi$ is a probabilistic lower bound of $\rho^\phi(\mathbf{S},0)$. We summarize this result next, and provide the closed-loop result in the next section.

\begin{mythm}[Open-loop quantitative encoding result]\label{theorem_quan}
    Let the conditions from Problem \ref{prob} hold. It holds that $\text{Prob}(\rho^\phi(\mathbf{S},0)\ge \bar{r}_{0|0}^\phi)\ge 1-\delta$. If $\bar{r}_{0|0}^\phi > 0$, then we have that $\text{Prob}(\mathbf{S}\models \phi)\ge 1-\delta$.
\end{mythm}

\proof{
    Recall that  $\text{Prob}(\rho^\mu(\mathbf{S},\tau) \geq \bar{r}^\mu_{\tau|0}, \forall (\tau,\pi^\mu)\in \{0, \dots, T_\phi\}\times \mathcal{P})\ge 1-\delta$ by Corollary \ref{col:predicate_quan}. Since the task $\phi$ is in positive normal form, which excludes negations, and due to the recursive structure of $\phi$ we can directly conclude that $\text{Prob}(\rho^\phi(\mathbf{S},0)\ge \bar{r}_{0|0}^\phi)\ge 1-\delta$. 
    If $\bar{r}_{0|0}^\phi > 0$, we can obviously conclude that $\text{Prob}(\mathbf{S}\models \phi)\ge 1-\delta$.
}

We emphasize that the previous result will allow us to directly optimize over the variable $\bar{r}_{0|0}^\phi$ to achieve $\bar{r}_{0|0}^\phi>0$, which then implies that $\text{Prob}(\mathbf{S}\models\phi)\ge 1-\delta$.

We note that the quantitative encoding allows for robustness considerations and provides a more fine-grained analysis. The qualitative encoding, on the other hand, is computationally more tractable.

\subsection{Predictive MIP-Based Control Synthesis}\label{subsec:loop}

For control synthesis, we can now use the qualitative encoding via the binary variable $\bar{z}_{0|k}^{\phi}$ or the quantitative encoding via the continuous variable $\bar{r}_{0|k}^\phi$ to solve Problem \ref{prob}.  Formally,  we then solve the following qualitative mixed integer optimization problem
\begin{subequations}\label{eq:opt}
	\begin{align}
		& \mathbf{u}^*_{k:T_\phi-1}:=\underset{\mathbf{u}_{k:T_\phi-1}}{\text{argmin}} & & J(x_k, \mathbf{u}_{k:T_\phi-1}) \\
		& \text{subject to} & & \nonumber \\
		& &&
		\!\!\!\!\!\!\!\!\!\!\!\!\!\!\!\!\!\!\!\!\!\!\!
		\mathbf{u}_{k:T_\phi-1} \in \mathcal{U}^{T_\phi-k}, \mathbf{x}_{k+1:T_\phi} \in \mathcal{X}^{T_\phi-k} \\
		& &&
		\!\!\!\!\!\!\!\!\!\!\!\!\!\!\!\!\!\!\!\!\!\!\!
		x_{\tau+1} = f(x_\tau, u_\tau), \tau=k, \dots, T_\phi-1, \\
		& &&
		\!\!\!\!\!\!\!\!\!\!\!\!\!\!\!\!\!\!\!\!\!\!\!
		\bar{z}_{0|k}^{\phi} = 1.  \label{eq:opt-stl-k-d}
	\end{align} 
\end{subequations}
Similarly, the quantitative mixed integer optimization problem can be formulated by replacing $\bar{z}_{0|k}^{\phi} = 1$ with $\bar{r}_{0|k}^{\phi} > 0$ in equation \eqref{eq:opt-stl-k-d}.

Note particularly that we enforce either $\bar{z}_{0|k}^{\phi} = 1$ or $\bar{r}_{0|k}^{\phi} > 0$ in equation \eqref{eq:opt-stl-k-d}. This enables us to select between the quantitative and the qualitative encoding.
We can also enforce $\bar{r}_{0|k}^{\phi} > a$, where $a>0$, to obtain a trajectory with higher robustness.

\textbf{Open-loop synthesis.} We will start with analyzing the case when $k=0$ where we get an open-loop control sequence $\mathbf{u}^*_{0:T_\phi-1}$. By Theorems \ref{theorem_qual} and \ref{theorem_quan} it is easy to see that we solve Problem \ref{prob} if we can find a feasible solution to the optimization problem in~\eqref{eq:opt}.  We next summarize the result that solves Problem \ref{prob}.

\begin{mythm}[Open-loop control]\label{thm:open}
    Let the conditions from Problem \ref{prob} hold. If the optimization problem in \eqref{eq:opt} is feasible at time $k=0$, then applying the open-loop control sequence $\mathbf{u}^*_{0:T_\phi-1}$ to \eqref{eq:sys} results in
    \begin{align}
        \text{Prob}(\mathbf{S} \models \phi) \geq 1-\delta. \nonumber
    \end{align}
\end{mythm}
\proof{
Note that we enforce that either $\bar{z}_{0|0}^{\phi} = 1$ or $\bar{r}_{0|0}^{\phi} > 0$ within the optimization problem in \eqref{eq:opt}. By Theorems \ref{theorem_qual} and \ref{theorem_quan} this implies that $\text{Prob}(\mathbf{S}\models \phi)\ge 1-\delta$ or  $\text{Prob}(\rho^\phi(\mathbf{S},0)\ge \bar{r}_{0|0}^\phi)\ge 1-\delta$, respectively. 
}

The complexity of the optimization problem in \eqref{eq:opt} for the open-loop is in general that of a mixed integer program which is NP-hard. 
Compared to a standard solution to deterministic STL control problems \cite{raman2014model}, our method introduces additional nonconvexity to the optimization framework. While we maintain identical constraints for Boolean and temporal operators, our formulation extends these constraints by incorporating KKT conditions for each predicate at every time step, which introduces the nonconvex constraint (\ref{eq:kkt}a) and complementary slackness condition (\ref{eq:kkt}d).

\section{Closed-Loop Predictive STL Control Synthesis}\label{sec:solution_closed}
The above open-loop control strategy may be conservative as the prediction regions $C_{\tau|0,i}$ will usually be conservative for large  $\tau$ as then the predictions $\hat{Y}_{\tau|0}$ lose accuracy. Furthermore, due to the lack of state feedback, the open-loop controller is not robust. Therefore, we propose to solve the optimization problem in \eqref{eq:opt} in a receding horizon fashion as in \cite{lindemann2023safe}. 

An intuitive solution to this problem is the MPC framework as follows.  At each time $k$ we observe the realization of $Y_{k}$, update our predictions $\hat{Y}_{\tau|k}$, and compute the control sequence $\mathbf{u}^*_{k:T_\phi-1}$ by solving the optimization problem in \eqref{eq:opt}.
Then, we apply the first element $u^*_{k}$ of $\mathbf{u}^*_{k:T_\phi-1}$ to the system in \eqref{eq:sys} before we repeat the process. 
However, since the STL task is defined over the whole trajectory $\mathbf{S}$ (instead of state constraints $c(x_k, Y_k) \geq 0$ in \cite{lindemann2023safe}), we cannot apply this strategy and get the same result as in \cite{lindemann2023safe}. This poses a new challenge that we address here.
Specifically, \cite{lindemann2023safe} show that 
$\text{Prob}(c(x_k, Y_k) \ge 0, \forall k \in \{1, \dots, T\}) \ge 1 - \delta$, where $T$ is the length of the task,
if the underlying optimization problem is feasible at all time steps. 
However, this reasoning cannot be applied to the STL case for the following reason. 
Our optimization problem \eqref{eq:opt} depends on both past realizations and future predictions, as opposed to \cite{lindemann2023safe}, 
where the optimization problem only depends on future predictions, but not past realizations.
Indeed, past realizations affect the satisfaction of the STL specification and can hence break the feasibility of the optimization problem.
The authors in \cite{stamouli2024recursively} build upon \cite{lindemann2023safe} and present a recursively feasibile shrinking horizon MPC that guarantees probabilistic satisfaction of state constraints.
However, this strategy cannot be applied to the STL case for the same reasons since feasibility is affected by past realizations.

We propose a closed-loop control framework for STL tasks that is recursively feasibility with high probability.
We use the qualitative and quantitative encodings presented in Sections \ref{subsec:qua} and \ref{subsec:quan}, respectively, and the optimization problem \eqref{eq:opt} presented in Section 4.3 where we modify the constraints imposed for predicates, while we keep the constraints for Boolean and temporal operators without modification. In order to change the constraints for predicates, we need new prediction regions which we present in Section \ref{subsec:closed_cp} before we discuss the qualitative and quantitative encodings in Section \ref{subsec:closed_stl} as well as the closed-loop controller in Section \ref{subsec:closed_control}.

\subsection{Closed-Loop Prediction Regions}\label{subsec:closed_cp}
We now apply conformal prediction to  capture the union of past prediction regions, which will allow us to design a shrinking-horizon MPC with probabilistic recursive feasibility guarantees.
Specifically, we are interested in constructing prediction regions for all prediction times $\tau \in \{1, \dots, T_\phi\}$ made at previous  times $s \in \{0, \dots, \tau-1\}$ for all agents $i \in \{1, \dots, N\}$, i.e., such that
\begin{align}\label{eq:cp_closed}
    \text{Prob}(||Y_{\tau,i} - \hat{Y}_{\tau|s,i}|| \leq C_{\tau|s, i}, \forall (\tau,s,i) \!\in\! \{1, \dots, T_{\phi}\} \times  \{0, \dots, \tau-1\} \times\{1, \dots, N\}) \!\geq\! 1-\delta,  
\end{align}
where $C_{\tau|s,i}$ indicates the prediction error of agent $i$ for predictions $\hat{Y}_{\tau|s,i}$ made at time $s$.
Intuitively, this guarantees that the probability of $Y_\tau$ being in the previous prediction regions described by $C_{\tau|s, i}$ for all previous times $s \in \{0, \dots, \tau-1\}$ and agents $i\in\{1,\hdots,N\}$ is larger than or equal to $1-\delta$, 
which will allow us to formulate constraints that guarantee feasibility of the optimization problem \eqref{eq:opt} (more details follow).

To obtain the prediction region in \eqref{eq:cp_closed}, we follow a similar idea to the open-loop case, but instead define the normalized nonconformity score 
\begin{align}\label{eq:nonscore_closed_}  
    & R_{CL}^{(j)}  \!\coloneqq\! \max_{\substack{(\tau, s, i) \in \{1,\dots, T_\phi\} \\ \times \{0, \dots, \tau-1\} \times \{1,\dots, N\}}}  \frac{||Y_{\tau,i}\!-\!\hat{Y}_{\tau|s, i}||}{\sigma_{\tau|s, i}} 
\end{align}
for calibration trajectories $Y^{(j)}\in D_{cal}$  where $\sigma_{\tau|s, i}>0$ are again constants that normalize the prediction error to  $[0,1]$ for all prediction times $\tau$, past times $s$ and agents $i$. In this case, we let $\sigma_{\tau|s, i}:=\max_{j} \|Y^{(j)}_{\tau,i} - \hat{Y}^{(j)}_{\tau|s,i}\|$ for training trajectories  $Y^{(j)}\in D_{train}$, again assuming that $\sigma_{\tau|s, i}$ is non-zero. 
Then, we have the following theorem.

\begin{mythm}\label{thm:prob_close}
    Given the random trajectory $Y\coloneqq (Y_0, Y_1, \dots) \sim \mathcal{D}$, a set of trajectories $D_{cal}$ and $D_{train}$ that satisfy Assumption \ref{ass:data}, and the failure probability $\delta \!\in\! (0,1)$, then the prediction region in equation \eqref{eq:cp_closed} holds for the choice of 
    \begin{align}\label{eq:C_closed_}
        C_{\tau|s,i} \coloneqq C_{CL} \sigma_{\tau|s, i} 
    \end{align}
    where $\sigma_{\tau|s, i}$ is a positive constant and
    \begin{align*}
        C_{CL} \coloneqq \text{Quantile}_{1-\delta}(R_{CL}^{(1)},\hdots,R_{CL}^{(K)}, \infty).
    \end{align*} 
\end{mythm}
The proof follows exactly the same steps as Theorem \ref{thm:proba} and is omitted.  Similarly, we can compute all $C_{\tau|s,i}$ in an offline manner by modifying Algorithm \ref{alg:offline}.

\subsection{Closed-Loop STL Encoding}\label{subsec:closed_stl}
Future predictions will affect recursive feasibility since the prediction regions made at time $k+1$ may be different from that made at time $k$, which makes the optimization problems \eqref{eq:opt} at time $k+1$ different from that at time $k$.
To account for this, inspired by \cite{stamouli2024recursively}, we gradually relax the constraints as now prediction regions become available online. Following this idea, we introduce qualitative and quantitative encodings in the remainder.

\textbf{Qualitative encoding.} We replace the quantitative encoding of predicates in equation \eqref{eq:predicate_} by 
\begin{align}\label{eq:predicate_closed}
    - \max_{0 \le s \le k}{\min_{y \in \mathcal{B}_{\tau|s}} \mu(x_{\tau|k}, y)}  \leq M (1-\bar{z}_{\tau|k}^{\mu}) -\epsilon,
\end{align}
where $\mathcal{B}_{\tau|s}$ is defined as before in equation \eqref{eq:B_tau}, but now with $C_{\tau|s,i}$ as in \eqref{eq:C_closed_}.
We use $s$ to denote the past time steps, and recall that $k$ is the current time and $\tau$ is the prediction time.
Equation \eqref{eq:predicate_closed} enforces $-\min_{y \in \mathcal{B}_{\tau|s}} \mu(x_{\tau|k}, y) \leq M (1-\bar{z}_{\tau|k}^{\mu}) -\epsilon$ to be satisfied for at least one of the prediction regions $\mathcal{B}_{\tau|s}, 0 \le s \le k$, instead of the current prediction region $\mathcal{B}_{\tau|k}$.
Intuitively, the max operation in \eqref{eq:predicate_closed} allows us to pick the least conservative prediction region to encode the predicate function, which relaxes the constraints gradually. 
Figure \ref{fig:illustration} illustrates the aforementioned idea by an example. Specifically, the predicate describes the task of two agents being close with distance of at most two, i.e., $\mu(x,y) = 2-||x-y||$.
The two grey areas indicate the feasible areas at different times, i.e., for the constraints $\min_{y \in \mathcal{B}_{2|0}} \mu(x, y) \ge 0$ and $\min_{y \in \mathcal{B}_{2|1}} \mu(x, y) \ge 0$. 
The maximum operator in equation \eqref{eq:predicate_closed} allows to choose the union of these areas as the feasible set at time step one.

\begin{figure}[h]
    \centering
    \begin{tikzpicture}[scale=1]
        \def\unit{1cm} 

        \fill[black!10] (0,0) circle (0.7*\unit);   
        \fill[black!10] (\unit,0) circle (0.5*\unit); 

        \draw[blue] (0,0) circle (\unit);  
        \draw[dashed] (0,0) circle (0.7*\unit); 
        
        \draw[dashed] (\unit,0) circle (0.5*\unit);  
        \draw[blue] (\unit,0) circle (1.2*\unit); 

        \draw[black] (-0.707*\unit, 0.707*\unit) -- (-1.2*\unit,0.9*\unit);
        \node at (-1.6*\unit,0.7*\unit) [above] {$\mathcal{B}_{2|1}$};
        \draw[black] (1.98*\unit, 0.7*\unit) -- (2.5*\unit,0.9*\unit);
        \node at (2.9*\unit,0.7*\unit) [above] {$\mathcal{B}_{2|0}$};
        
        \draw[black] (-0.3*\unit, -0.4*\unit) -- (-1.2*\unit,-0.9*\unit);
        \node at (-2.2*\unit,-1.4*\unit) [above] {\tiny $\min_{y \in \mathcal{B}_{2|1}} \mu (x_{2|1}, y) \ge 0$};
        \draw[black] (1*\unit, -0.3*\unit) -- (2.2*\unit,-0.9*\unit);
        \node at (3.2*\unit,-1.4*\unit) [above] {\tiny $\min_{y \in \mathcal{B}_{2|0}} \mu (x_{2|1}, y) \ge 0$};

    \end{tikzpicture}
    \caption{Illustration of the idea of gradual constraint relaxation.}
    \label{fig:illustration}
\end{figure}
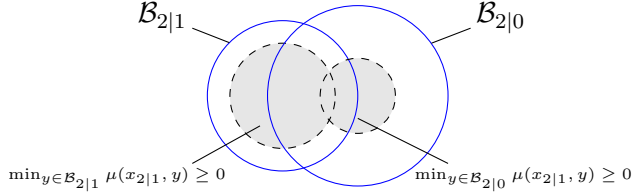

\textbf{Computational considerations. }
To obtain computational tractability for \eqref{eq:predicate_closed}, we introduce a binary variable for each $s$ as $q_{\tau|k}^{\mu, s}$ to indicate whether the predicate is satisfied under $\mathcal{B}_{\tau|s}$, i.e., $q_{\tau|k}^{\mu, s} = 1$ iff $- \min_{y \in \mathcal{B}_{\tau|s}} \mu(x_{\tau|k}, y) \leq M (1-\bar{z}_{\tau|k}^{\mu}) -\epsilon$.
Specifically, we have the following constraints for each $s, 0 \le s \le k$:
\begin{subequations}\label{eq:closed_reformulated}
    \begin{align}
        \!\!\! - \!\!\! \min_{y \in \mathcal{B}_{\tau|s}} \!\!\! \mu(x_{\tau|k}, y)  \!-\! M(1 \!-\! \bar{z}_{\tau|k}^{\mu}) \!+\! \epsilon & \!\leq\! M' (1 \! - \! q_{\tau|k}^{\mu, s})  \! - \! \epsilon' \! , \!\!\!\! \label{eq:closed_reformulated1}  \\
        \!\!\!\! \min_{y \in \mathcal{B}_{\tau|s}} \!\!\! \mu(x_{\tau|k}, y) \!+\! M(1 \! - \! \bar{z}_{\tau|k}^{\mu}) \!-\! \epsilon  & \!\leq\! M' q_{\tau|k}^{\mu, s}  \! - \! \epsilon', \label{eq:closed_reformulated2} 
    \end{align} 
\end{subequations}
where $M'$ and $\epsilon'$ are sufficiently large and sufficiently small positive constants similar to $M$ and $\epsilon$. For simplicity, we can choose $M'=M$ and $\epsilon' = \epsilon$.
Then, instead of enforcing \eqref{eq:predicate_closed}, we enforce \eqref{eq:closed_reformulated} together with 
\begin{align}
    \sum_{s=0}^{k} q_{\tau|k}^{\mu, s} & \ge 1. \label{eq:closed_reformulated3}
\end{align}
By enforcing \eqref{eq:closed_reformulated3}, it holds that $q_{\tau|k}^{\mu, s}=1$ for at least one $s \in \{0, \dots, k\}$, so that $- \min_{y \in \mathcal{B}_{\tau|s}}$ $\mu(x_{\tau|k}, y) \leq M (1-\bar{z}_{\tau|k}^{\mu})  -\epsilon$ holds for at least one of $\mathcal{B}_{\tau|s}$, which implies that \eqref{eq:predicate_closed} holds. In the opposite direction, if $- \min_{y \in \mathcal{B}_{\tau|s}}$ $\mu(x_{\tau|k}, y) \leq M (1-\bar{z}_{\tau|k}^{\mu})  -\epsilon$ holds under $\mathcal{B}_{\tau|s}$, then $q_{\tau|k}^{\mu, s} = 1$ and \eqref{eq:closed_reformulated3} holds. 
As a result, we can use equations \eqref{eq:closed_reformulated} and \eqref{eq:closed_reformulated3} to replace \eqref{eq:predicate_closed}.

\textbf{Summary of qualitative encoding.} To encode $\bar{z}_{0|k}^{\phi}$ in the optimization problem \eqref{eq:opt}, we encode predicates by equations \eqref{eq:closed_reformulated} and \eqref{eq:closed_reformulated3} instead of equation \eqref{eq:predicate_}.
Additionally, we encode Boolean and temporal operators in the same way as in the open-loop case in Section \ref{subsec:qua}.

\textbf{Quantitative encoding.}
Similarly, we encode predicates by gradually relaxed constraints for the quantitative encoding.
Specifically, we replace the quantitative encoding of predicates in equation \eqref{eq:rho_bar} by
\begin{align}\label{eq:rho_bar_close}
    \!\!\!\!\! \bar{r}_{\tau|k}^{\mu} \!\coloneqq \!
    \left\{
        \begin{array}{ll}
        \mu(x_\tau, Y_\tau) \ &\text{if} \ \tau \leq k, \\
        \max_{0 \le s \le k} \min_{y \in \mathcal{B}_{\tau|s}} \mu(x_{\tau|k}, y) & \text{otherwise.}
        \end{array}
    \right. 
\end{align} 
Intuitively, instead of considering the worst case of $\mu(x_{\tau|k}, y)$ for $y$ within the prediction region $\mathcal{B}_{\tau|k}$, we choose the least conservative prediction region from $\mathcal{B}_{\tau|s}, 0 \le s \le k$ to encode the predicate by the maximum operator.

\textbf{Computational considerations. }
To obtain computational tractability for the maximum operator when $\tau > k$ in equation \eqref{eq:rho_bar_close}, 
we apply a similar idea as for the Boolean encoding of the disjunction operator by introducing the binary variable $q_{\tau|k}^{\mu, s}$ to indicate which value is the largest. 
Specifically, we use the following constraints to replace \eqref{eq:rho_bar_close} in the case that $\tau > k$:
\begin{subequations}\label{eq:quan_reformulated}
    \begin{align}
        & \!\!\!\Sigma_{s=0}^k q_{\tau|k}^{\mu, s} = 1, \\
        & \!\!\!\bar{r}_{\tau|k}^{\mu} \geq \min_{y \in \mathcal{B}_{\tau|s}} \mu(x_{\tau|k}, y), s=0,\dots, k, \\
        & \!\!\! \bar{r}_{\tau|k}^{\phi} \leq \min_{y \in \mathcal{B}_{\tau|s}} \mu(x_{\tau|k}, y) \!+\! M(1 \!-\! q_{\tau|k}^{\mu, s}), s \!=\! 0,\dots, k, \!\! \\
        & \!\!\! \bar{r}_{\tau|k}^{\phi} \geq \min_{y \in \mathcal{B}_{\tau|s}} \mu(x_{\tau|k}, y) \!-\! (1 \!-\! q_{\tau|k}^{\mu, s}) M, s \!=\! 0,\dots, k. \!\!
    \end{align}
\end{subequations}
Consequently, $q_{\tau|k}^{\mu, s} = 1$ iff $\bar{r}_{\tau|k}^{\mu} = \min_{y \in \mathcal{B}_{\tau|s}} \mu(x_{\tau|k}, y)$.

\textbf{Summary of quantitative encoding.}  To encode $\bar{r}_{0|k}^{\phi}$ in the optimization problem \eqref{eq:opt}, we encode predicates by equations \eqref{eq:rho_bar_close} and \eqref{eq:quan_reformulated} instead of equation \eqref{eq:rho_bar}.
We again encode Boolean and temporal operators in the same way as in the closed-loop case in Section \ref{subsec:quan}.

The complexity of the optimization problem in \eqref{eq:opt} for the closed-loop is also in general that of a mixed integer program which is NP-hard. 
Compared to a standard solution to deterministic STL control problems \cite{raman2014model}, our closed-loop optimization formulation at time $k$ introduces more constraints due to the explicit consideration of all previously predicted regions. Specifically, this approach generates approximately $k+1$ times more nonconvex constraints than that we discussed at the end of Section \ref{sec:solution_open}.

\subsection{Closed-Loop Predictive MIP-Based Control Synthesis}\label{subsec:closed_control}

In this subsection, we present our closed-loop control framework.
At each time $k$, we observe the realization of $Y_{k}$, compute the control sequence $\mathbf{u}^*_{k:T_\phi-1} \coloneqq u_{k|k}^*,\dots, u_{T_\phi-1|k}^*$ based on the previous prediction regions $\mathcal{B}_{\tau|s}, 0 \le s \le k$, by solving the optimization problem in \eqref{eq:opt}, where $\bar{z}_{0|k}^{\phi}$ and $\bar{r}_{0|k}^{\phi}$ are encoded as we described in the last subsection.
Then, we apply the first element $u^*_{k}$ of $\mathbf{u}^*_{k:T_\phi-1}$ to the system in \eqref{eq:sys} before we repeat the process. Naturally, we do so in a shrinking horizon manner where the prediction horizon decreases by one at each time. 

\begin{mythm}[Closed-loop control]\label{thm:close1}
    Let the conditions from Problem \ref{prob} hold. Suppose the optimization problem in \eqref{eq:opt} is feasible at the initial time $k=0$, where the predicates are encoded as in \eqref{eq:closed_reformulated} and \eqref{eq:closed_reformulated3} (qualitative encoding) or \eqref{eq:quan_reformulated} (quantitative encoding). 
    Then, the probability of \eqref{eq:opt} being feasible at every time step $k \in\{ 1, \dots, T_\phi-1\}$ is larger than or equal to $1-\delta$.
    Furthermore, applying $u_{k|k}^*$ results in 
    $\text{Prob}(\mathbf{S} \models \phi) \geq 1-\delta.$
\end{mythm}

\proof{
    For brevity, we only show the qualitative case, while the proof for the quantitative case follows exactly the same steps.
    At time $k$, assume that  we have obtained the optimal solution $\mathbf{u}^*_{k:T_\phi-1} \coloneqq u_{k|k}^*,\dots, u_{T_\phi-1|k}^*$ for the optimization problem \eqref{eq:opt} such that $\bar{z}_{0|k}^\phi = 1$.
    This means that there is an assignment of $0$ and $1$ values of the binary variables $\bar{z}_{\tau|k}^\mu$ for each pair $(\tau, \pi^\mu)\in\{0,\hdots,T_\phi\} \times \mathcal{P}$. We note that  this assignment is potentially non-unique. 
    Specifically, let $\bar{z}_{\tau|k}^\mu = 1$ for all pairs $(\tau, \pi^\mu) \in TP_k$, 
    where $TP_k \subseteq \{1, \dots, T_\phi\}\times \mathcal{P}$, and let $\bar{z}_{\tau|k}^\mu = 0$ otherwise. 
    For each pair $(\tau, \pi^\mu) \in TP_k$ such that $\bar{z}_{\tau|k}^\mu = 1$, we further know by equations \eqref{eq:closed_reformulated} and \eqref{eq:closed_reformulated3} that there exists at least one $s \in \{0, \dots, k\}$ such that $q_{\tau|k}^{\mu, s} = 1$.
    If it holds that $Y_{k+1} \in \mathcal{B}_{k+1|s}$ for all $s \in \{0, \dots, k\}$, which holds with probability $1-\delta$ due to Theorem \ref{thm:prob_close}, then we have a feasible solution $\mathbf{u}_{k+1} \coloneqq u_{k+1|k}^*,\dots, u_{T_\phi-1|k}^*$ of the optimization problem \eqref{eq:opt} at the next time step $k+1$. This holds for the following reasons:
    \begin{itemize}
        \item For $\tau \leq k$, we have $\bar{z}_{\tau|k+1}^\mu = \bar{z}_{\tau|k}^\mu = 1$ for all $(\tau, \pi^\mu) \in TP_k$ since the states $x_\tau$ and $Y_\tau$ were known already; \vspace{-6pt}
        \item For $\tau = k+1$, we have $\bar{z}_{\tau|k+1}^\mu = \bar{z}_{\tau|k}^\mu = 1$ for all $(\tau, \pi^\mu) \in TP_k$ due to the minimum operation in \eqref{eq:closed_reformulated} if $Y_{k+1} \in \mathcal{B}_{k+1|s}$ for all $s \in \{0, \dots, k\}$, which again holds with probability $1-\delta$ due to Theorem \ref{thm:prob_close}; \vspace{-6pt}
        \item For $\tau \geq k+2$, we have $\bar{z}_{\tau|k+1}^\mu = \bar{z}_{\tau|k}^\mu = 1$ for all $(\tau, \pi^\mu) \in TP_k$ since 
        \eqref{eq:closed_reformulated} and \eqref{eq:closed_reformulated3} hold at time step $k+1$ for $x_{\tau|k+1} = x_{\tau|k}$ (obtained by applying the controller $\mathbf{u}_{k+1}$) and for $q_{\tau|k+1}^{\mu, s} = q_{\tau|k}^{\mu, s}$.
    \end{itemize}
    
    Using the fact that $\phi$ is in  positive normal form, these observations imply that $\mathbf{u}_{k+1}$ is a feasible solution of the optimization problem \eqref{eq:opt} with $\bar{z}_{0|k+1}^\phi = 1$ at time $k+1$  if $\bar{z}_{0|k}^\phi = 1$ and if  $Y_{k+1} \in \mathcal{B}_{k+1|s}$ for all $s \in \{0, \dots, k\}$, which holds with probability $1-\delta$ due to Theorem \ref{thm:prob_close}.

    We can repeat the previous reasoning for all $T_\phi$ time steps, i.e., we can ensure recursive feasibility at each time  $k\in\{1, \dots, T_\phi-1\}$ if it holds that  $Y_{k} \in \mathcal{B}_{k|s}$ for all $k \in \{1, \dots, T_\phi-1\}$ and $s \in \{0, \dots, k-1\}$. Due Theorem \ref{thm:prob_close}, we thus know that the probability of \eqref{eq:opt} being feasible at every time step $k \in\{ 1, \dots, T_\phi-1\}$ is larger than or equal to $1-\delta$, which completes the first part of the proof.

    At time step $k = T_\phi-1$, $\bar{z}_{0|k}^\phi = 1$ implies $\bar{z}_{0|k+1}^\phi = 1$ with probability $1-\delta$ due to the recursive feasibility as above. This means that the trajectory $s_0 s_1 \dots s_{T_\phi}$ will satisfy $\phi$ with probability $1-\delta$, which completes the second part of the proof - probabilistic task satisfaction.
}

The probabilistic nature of recursive feasibility in our framework arises from the uncertainty in predicting the state $Y_{k+1}$ of uncontrollable agents at time $k$. Since $Y_{k+1}$ is uncontrollable and unbounded, we cannot provide deterministic guarantees. Instead, we rely on probabilistic statements to ensure that recursive feasibility holds with high probability. Recent literature \cite{fiacchini2024measured} has also adopted probabilistic statements to describe recursive feasibility in the presence of unbounded stochastic uncertainty. This probabilistic guarantee is particularly relevant in practical applications where the behavior of uncontrollable agents (e.g., pedestrians, other vehicles) cannot be deterministically bounded.

\section{Case studies}\label{sec:case}

In this section, we illustrate our control method on two case studies.  All simulations are conducted in \textsf{Python 3}  and we use \textsf{SCIP} \cite{bestuzheva2023enabling} to solve the optimization problem. 
Our implementations are available at \url{https://github.com/Xinyi-Yu/STL-Synthesis-among-Uncontrollable-Agents}, where more details can be found.

\subsection{Building temperature control}\label{subsec:case1}

\begin{figure}[t]
	\centering
	\input{floor_plan}
	\caption{The floor plan of the hall.}
	\label{fig:layout}
\end{figure}
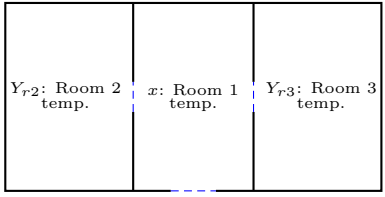

\begin{figure}
    \centering
    \includegraphics[width = 0.7\linewidth]{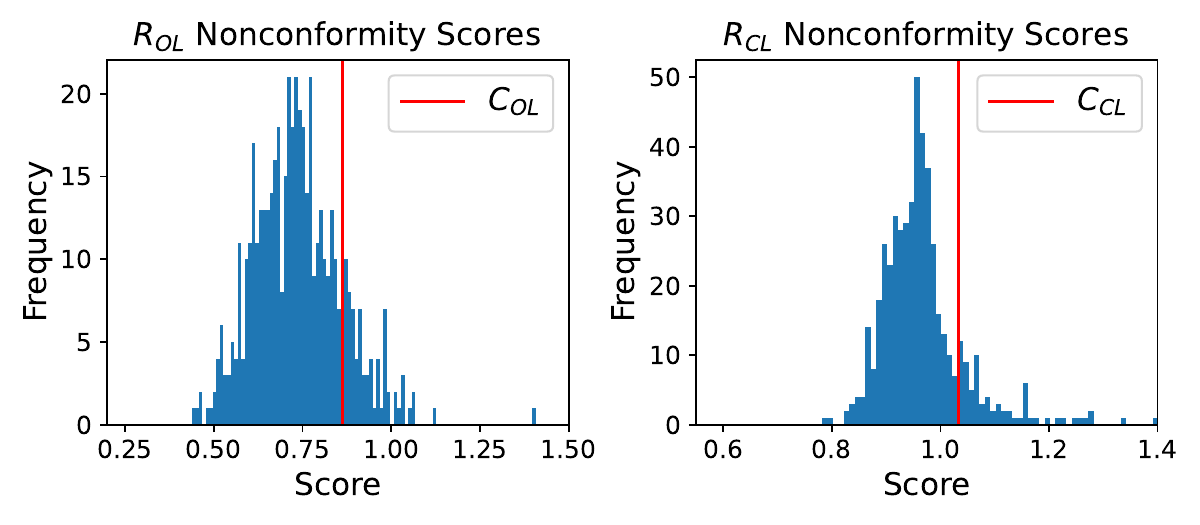}
    \caption{Nonconformity scores $R_{OL}^{(j)}$ (left) and $R_{CL}^{(j)}$ (right) on $D_{cal}$ in the temperature control case.}
    \label{fig:case1_his}
\end{figure}

\begin{figure}
    \centering
    \includegraphics[width = 0.7\linewidth]{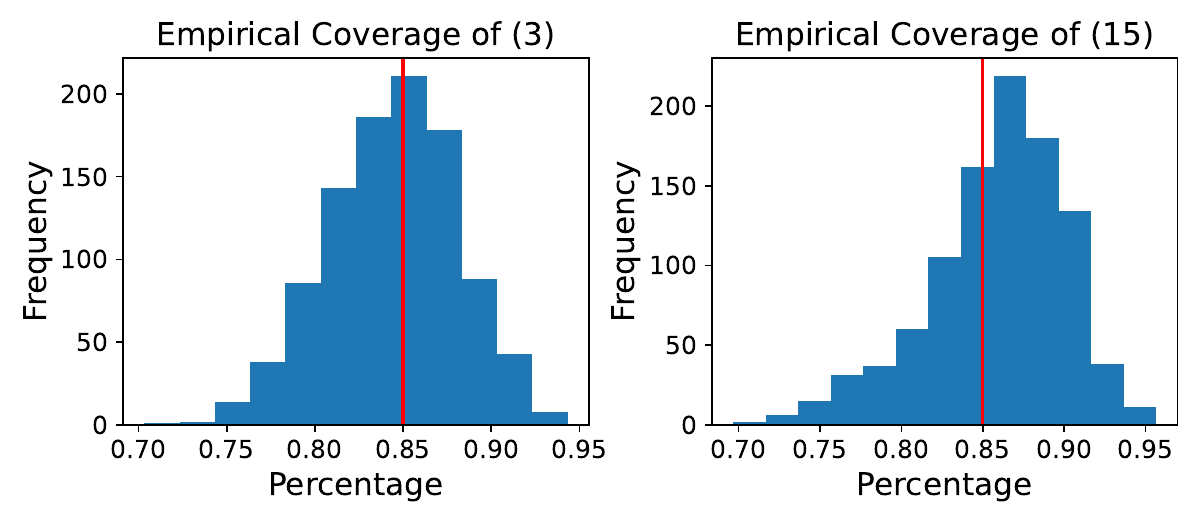}
    \caption{Empirical validation results about the coverage of $R_{OL}^{(j)} \leq C_{OL}$ (left) and $R_{CL}^{(j)} \leq C_{CL}$ (right) in the temperature control case.}
    \label{fig:case1_his_emp}
\end{figure}

\begin{figure*}
    \centering
    \includegraphics[width = 1\linewidth]{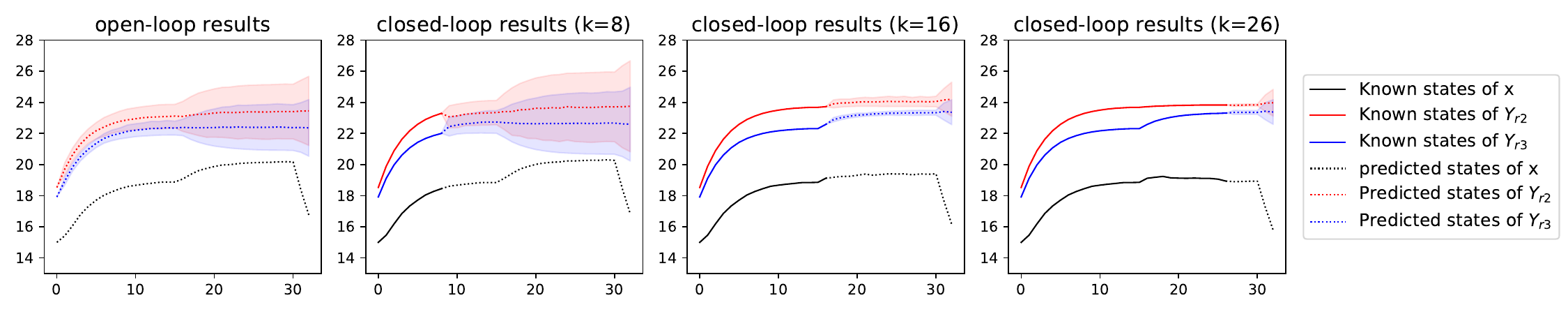}
    \caption{The result of the temperature qualitative control of the first case study.}
    \label{fig:case1_trajectory_qualitative}
\end{figure*}

\begin{figure*}
    \centering
    \includegraphics[width = 1\linewidth]{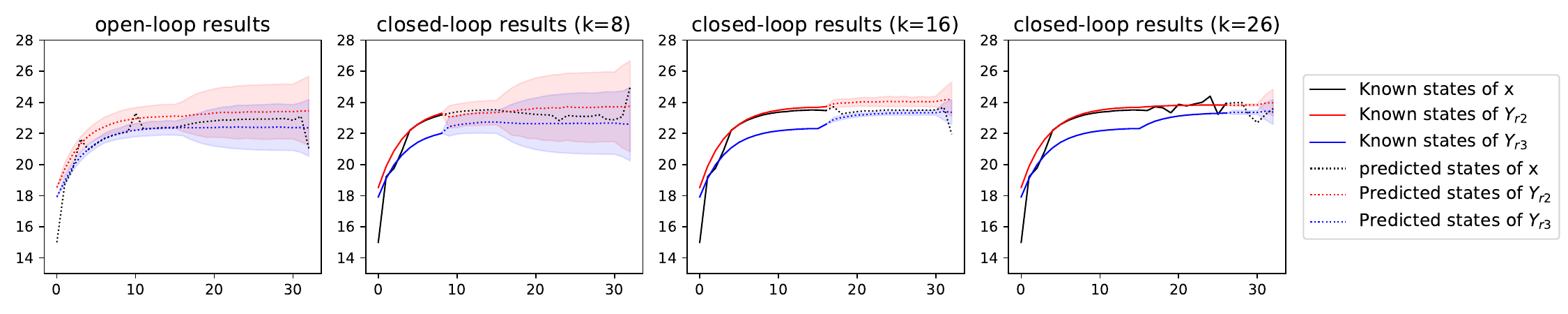}
    \caption{The result of the temperature quantitative control of the first case study.}
    \label{fig:case1_trajectory_quantitative}
\end{figure*}

We consider the temperature control of a public space in a hall with floor plan shown in  Figure~\ref{fig:layout}. Specifically, we can control the temperature of Room 1, which is the public space, while Rooms 2 and 3 are uncontrollable for us, e.g., these may be meeting rooms where the temperature can only be adjusted manually. Formally, the dynamics of the system are given as
\[
	x_{k+1} = x_k + \tau_s(\alpha_e(T_e - x_k) + \alpha_H(T_h - x_k )u_k),
\]
where the state $x_k \!\in\! \mathcal{X}=[0,45]$ denotes the temperature of the public space at time $k$, the control input $u_k \!\in\! \mathcal{U} = [0,1]$ is the ratio of the heater valve, $\tau_s$ is the sampling time for the discrete dynamics, $T_h = 55^\circ C$ is the heater temperature, $T_e = 5^\circ C$ is the outside temperature, and $\alpha_e = 0.06$ and $\alpha_H = 0.08$ are  heat exchange coefficients. The initial state is $x_0 = 5$. The model is adopted from \cite{jagtap2020formal}.

\textbf{System Specification.}
The task is to ensure that the difference between the temperatures in the public space and the two meeting rooms is bounded, e.g., such that during transition between rooms people feel comfortable. Let $Y_{r2}$ and $Y_{r3}$ denote the uncontrollable random variables describing the temperatures of Rooms 2 and 3.
In particular, we require that, 
the difference between $x$ and the temperatures of $Y_{r2}$ and $Y_{r3}$ should  be below a threshold of $5^\circ C$ for at least an hour starting from within the first $4$ minutes.
We set the sampling time to be  $\tau_s = 2$ minutes, and we describe the task as follows:
\begin{align*}
    \phi \!\coloneqq\! \mathbf{F}_{[0,2]} \mathbf{G}_{[0,30]} \big( x\!-\!Y_{r2} \!\leq\! 5 \wedge x\!-\!Y_{r2} \!\geq\! -5 \wedge x\!-\!Y_{r3} \!\leq\! 5 \wedge x\!-\!Y_{r3} \!\geq\! -5  \big). 
\end{align*}
In this case, we set the failure probability to be 15\%, i.e., $\delta := 0.15$.

\textbf{Behavior of uncontrollable rooms.} The group of people will adjust the  temperature in the rooms based on their individual temperature preferences when they enter the meeting room. We assume the temperature variation in $Y_{r2}$ and $Y_{r3}$  generated by this adjustment is described by Newton's Law of Cooling \cite{vollmer2009newton}.

\textbf{Data Collection.} We collected $2000$ trajectories of temperature trajectories for Rooms 2 and 3 -- generated as described before. We split them into training (used for the trajectory predictor), calibration (used for conformal prediction), and test (used for validation) datasets with sizes $|D_{train}|=500, |D_{cal}|=500$ and $|D_{test}|=1000$, respectively. From the training data $D_{train}$, we train a long-short-term memory (LSTM) network to make predictions about the room temperature. We note that we warm-start the LSTM at time $k=0$ with data, e.g., the LSTM is fed with past information $y_{-6}, \dots, y_{-1}$.

\textbf{Multi-agent prediction regions. } Fig.~\ref{fig:case1_his} shows histograms of the nonconformity scores $R_{OL}^{(j)}$ and $R_{CL}^{(j)}$ evaluated on $D_{cal}$. From these histograms of nonconformity scores, we compute $C_{OL} := 0.863$ and $C_{CL} := 1.033$ according to Theorems \ref{thm:proba} and \ref{thm:prob_close}. 
Next, we empirically validate the correctness of the prediction regions in \eqref{eq:cp_open_} and \eqref{eq:cp_closed}. 
We perform the following two experiments $1000$ times each. 
In the first experiment, we randomly sample $150$ calibration trajectories from $D_{cal}$ and $150$ test trajectories from $D_{test}$. Then, we construct $C_{OL}$ and $C_{CL}$ from the $150$ calibration trajectories and compute the ratio of how many of the $150$ test trajectories satisfy $R_{OL}^{(j)} \leq C_{OL}$ and $R_{CL}^{(j)} \leq C_{CL}$, respectively.
Fig.~\ref{fig:case1_his_emp} shows the histogram over these ratios, and we can observe that the result achieves the desired coverage $1-\delta=0.85$.
In the second experiment, we randomly sample $150$ calibration trajectories from $D_{cal}$ and $1$ test trajectories from $D_{test}$. Then, we construct $C_{OL}$ and $C_{CL}$ from the $150$ calibration trajectories and check whether or not the sampled test data satisfies $R_{OL}^{(j)} \leq C_{OL}$ and $R_{CL}^{(j)} \leq C_{CL}$, respectively.
We find a ratio of $0.913$ for $R_{OL}^{(j)} \leq C_{OL}$ and a ratio of $0.871$ for $R_{CL}^{(j)} \leq C_{CL}$, respectively, which further empirically confirms \eqref{eq:cp_open_} and \eqref{eq:cp_closed}.
For one of these test trajectories, Fig.~\ref{fig:case1_trajectory_qualitative} shows the prediction regions by shaded areas defined by $\hat{Y}_{\tau|k}$ and $C_{\tau|k}$ for $\tau > k$ where $k=0$, $k=8$, $k=16$ and $k=26$ from the left to the right, respectively. 

\textbf{Prediction region comparison.} We compare our method with \cite{cleaveland2024conformal}, where a parameterized prediction error was proposed. Instead of computing the maximum prediction error $\sigma$ as in \eqref{eq:nonscore_open} and \eqref{eq:nonscore_closed_}, they proposed an optimization-based method to obtain an in some sense optimal $\sigma$ that results in small prediction regions (see \cite{cleaveland2024conformal} for details).
Specifically, they formulate the problem as a mixed integer linear complementarity program (MILCP).
We extend their framework to the multi-agent case and solve the optimization problem with CasADI \cite{andersson2019casadi}.
The result is shown in Fig.~\ref{fig:case1_comp}.
Our computation time is 0.007s, much less than theirs, which is 2.82s.  Furthermore, the shape of the region looks similar. It is worth mentioning that although they claim they can find the smallest prediction region theoretically, it is sometimes difficult to realize this in practice due to nonconvexity of the optimization problem. 

\begin{figure}
    \centering
    \includegraphics[width = 0.7\linewidth]{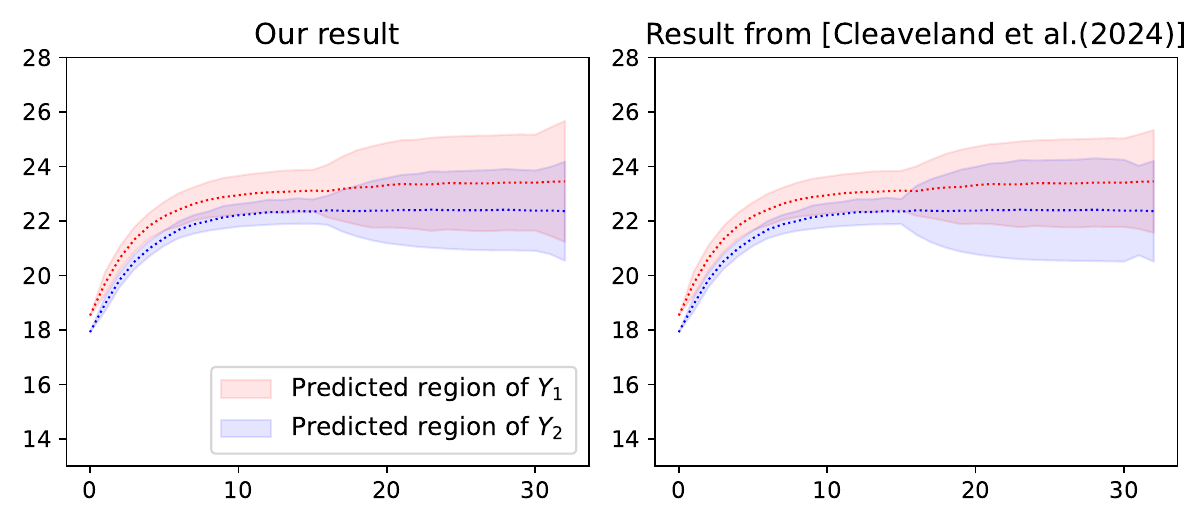}
    \caption{The comparison result of the open-loop prediction regions of the first case study.}
    \label{fig:case1_comp}
\end{figure}

\textbf{Qualitative predictive control synthesis.} 
We run the open-loop and closed-loop control frameworks $1000$ times with data from $D_{test}$, respectively, where the cost function is set to be the minimum summation of squared control inputs as $J \coloneqq \sum_{k=0}^{31} u_k^2$.
In the open-loop control case, the problem is feasible for 1000 times, where the average computation time, average robustness, and task satisfaction rate\footnote{The task satisfaction rate is computed as the ratio between the number of times $\phi$ is satisfied and the number of times the optimization problem is feasible at time $k=0$.} are 0.148s, 0.15, and 0.955, respectively.
In the closed-loop control case, the problem is feasible at the initial time $k=0$ for $999$ times\footnote{We note the difference in initial feasibility compared to the open-loop controller, which is caused since  prediction regions $\mathcal{B}_{\tau|0}$ of the closed-loop controller are more conservative.}, and recursive feasibility is achieved for $986$ times, where the average total computation time, average robustness, and task satisfaction rate are 5.90s, 0.07, and 0.987, respectively.
We remark that the closed-loop controller has slightly lower robustness than the open-loop controller. This is since the closed-loop prediction regions $\mathcal{B}_{\tau|k}$ for large $k$ appear to be less conservative than open-loop prediction regions. We also see that the closed-loop controller has a higher task satisfaction rate since the controller can update the control inputs in a receding horizon fashion to observations of the uncontrollable agents.
Fig.~\ref{fig:case1_trajectory_qualitative} shows one of the results of the proposed control framework with the qualitative encoding. Fig.~\ref{fig:case1_trajectory_qualitative}(a) specifically presents the result of the open loop controller where the prediction regions are large for large $\tau$, while Figs. \ref{fig:case1_trajectory_qualitative}(b)(c)(d) present the result of the closed loop controller at times $k=8, 16, 26$, respectively.

\textbf{Quantitative predictive control synthesis.} 
Similar to the qualitative control, we run the open-loop and closed-loop control frameworks with quantitative encoding $1000$ times with data in $D_{test}$, respectively, where we enforce the robust semantics to be larger than 1, i.e., $\bar{r}_{0|k}^\phi \geq 0$, and we maximize $\bar{r}_{0|k}^\phi$ in the cost function, i.e., we set $J \coloneqq - \bar{r}_{0|k}^{\phi}$.
In the open-loop control case, the problem is feasible for 927 times, where the average computation time, average robustness, and task satisfaction rate are 1.79s, 2.92, and 1, respectively.
In the closed-loop control case, the problem is feasible at the initial time $k=0$ for 923 times, and the recursive feasibility is achieved for 890 times, where the average computation time, average robustness, and task satisfaction rate are 147.29s, 2.80, and 0.964, respectively.
The closed-loop controller has a lower task satisfaction rate than the open-loop controller due to the nonconvexity of the optimization problem and performance issues of our solver in practice.
It is worthy noting that the computation time of the quantitative encoding is greater than that of the qualitative encoding, although it may achieve better robustness performance.
As we can see in the result in Fig.~\ref{fig:case1_trajectory_quantitative}, the trajectory stays near the center of the prediction regions, leading to higher robustness.

\subsection{Robot motion planning}

\begin{figure}
    \centering
    \includegraphics[width = 0.7\linewidth]{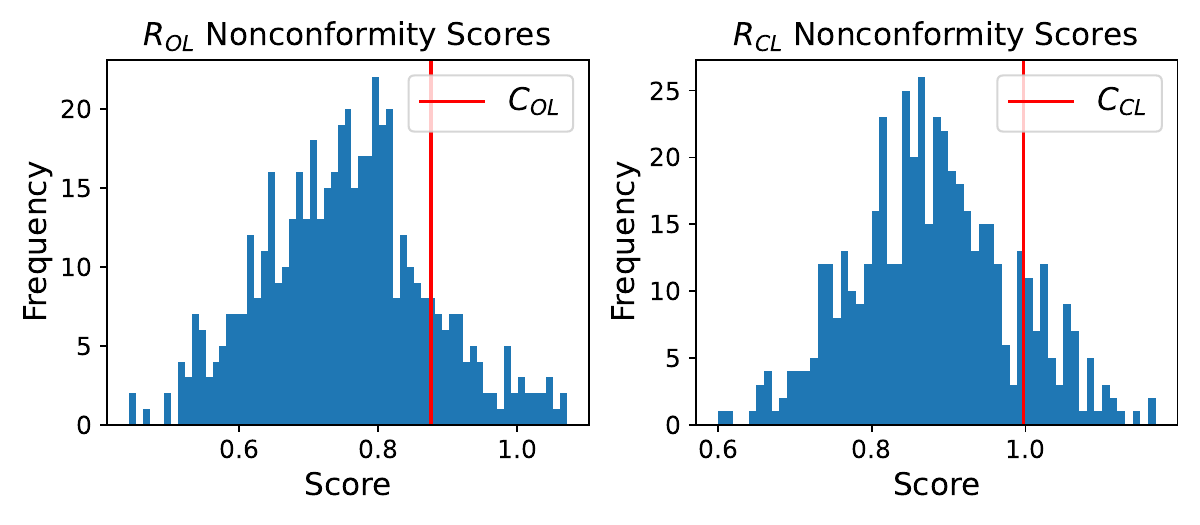}
    \caption{Nonconformity scores $R_{OL}^{(j)}$ (left) and $R_{CL}^{(j)}$ (right) on $D_{cal}$ in the robot motion planning case.}
    \label{fig:case2_his}
\end{figure}

\begin{figure}
    \centering
    \includegraphics[width = 0.7\linewidth]{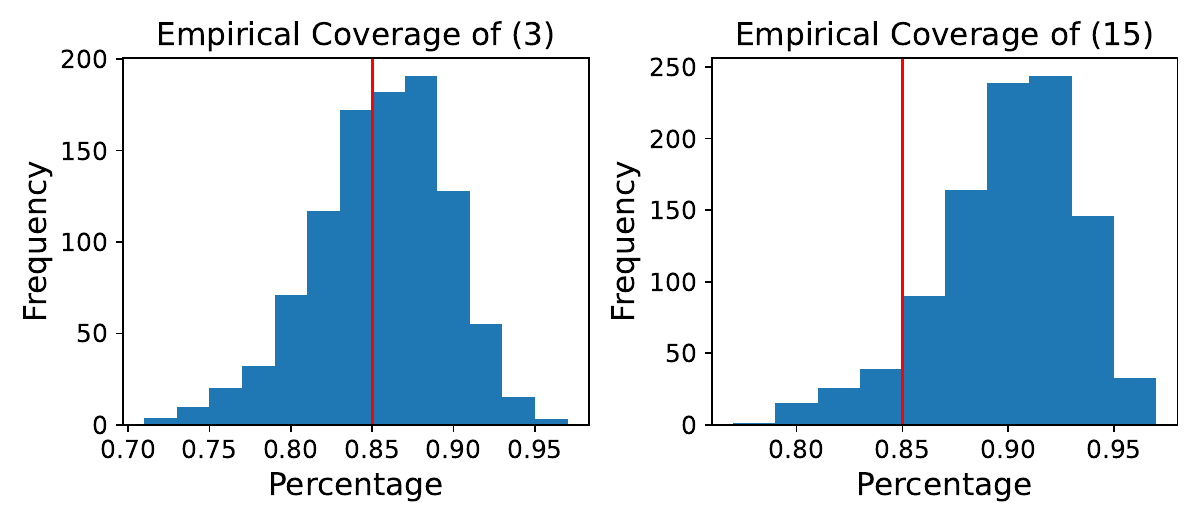}
    \caption{Empirical validation results about the coverage of $R_{OL}^{(j)} \leq C_{OL}$ (left) and $R_{CL}^{(j)} \leq C_{CL}$ (right) in the robot motion planning case.}
    \label{fig:case2_his_emp}
\end{figure}

\begin{figure*}
    \centering
    \includegraphics[width = 1\linewidth]{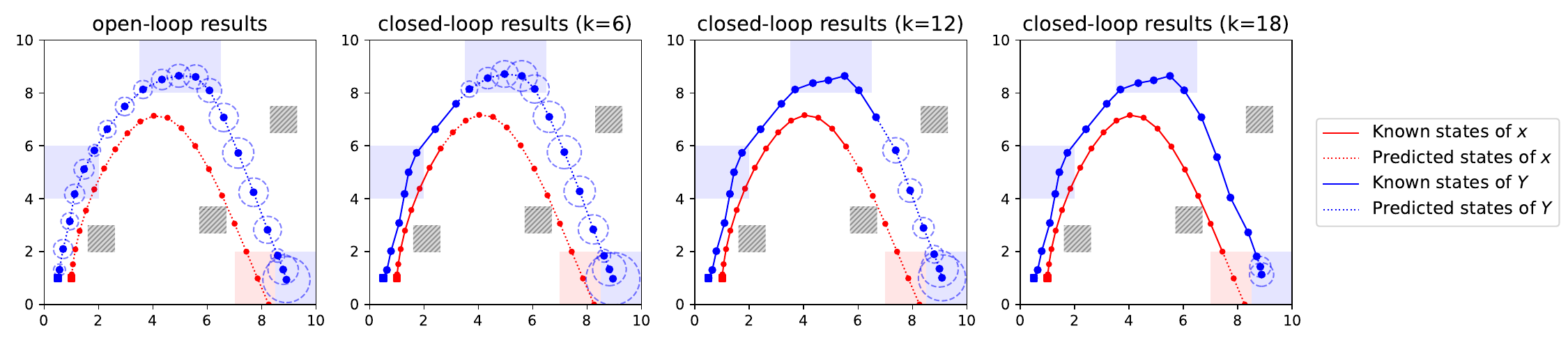}
    \caption{The result of the robot motion planning of the second case study.}
    \label{fig:case2_trajectory}
\end{figure*}

We consider the planar motion of a single robot (denoted by $R_1$) with double integrator dynamics. The  system model, with a sampling period of 1 second, is described as
\[
	x_{k+1} = 
	\begin{bmatrix}
		1 & 1 & 0 & 0 \\
		0 & 1 &  0 & 0 \\
		0 & 0 & 1 & 1 \\
		0 & 0 & 0 & 1 
	\end{bmatrix}
	x_{k} + 
	\begin{bmatrix} 
		0.5 & 0   \\ 
		1 & 0 \\ 
		0 & 0.5 \\ 
		0 & 1 
	\end{bmatrix}
	u_k,
\]
where the state $x_k = [p_x^1 \ v_x^1 \ p_y^1 \ v_y^1]^\top$ denotes the $x$-position, $x$-velocity, $y$-position and $y$-velocity, and where the control input $u_k = [u_x \ u_y]^T$ denotes the $x$-acceleration and $y$-acceleration, respectively.
The physical constraints are $x \!\in\! \mathcal{X} = [0, 10] \times [-1.5, 1.5] \times [0,10] \times [-1.5, 1.5]$ and $u \!\in\! \mathcal{U} = [-1, 1] \times [-1,1]$.
The initial state is $x_0 = [1, 0, 1, 0]^\top$.
Furthermore, we have one more uncontrollable agent (denoted by $R_2$) and its state  $Y_k = [P_x^2 \ P_y^2]^\top$  denotes the $x$-position and $y$-position.
Then we define the joint state as $s_k = [p_x^1 \ v_x^1 \ p_y^1 \ v_y^1 \ P_x^2 \ P_y^2]^\top$.
In this case, $R_1$ is controllable and $R_2$ is uncontrollable.

\textbf{System Specification.}
The objective of the two robots is to visit specific regions to perform specific tasks, such as collecting packages and transporting them to a final destination. For this collaborative task, $R_2$ is the leader and $R_1$ should follow $R_2$.
Specifically, $R_2$ should arrive and stay in the left blue region and top blue region, as shown in the Fig.~\ref{fig:case2_trajectory}, within the time intervals $[4, 6]$ and $[9,13]$, respectively. 
After that, starting from a certain time instant between the time interval $[16, 18]$, $R_1$ and $R_2$ should stay in bottom right red and blue regions respectively for at least 2 time instants.
During the task execution, the distances between the two agents should be at most $D$ meters such that the two robots can communicate or perform collaborative tasks.
Furthermore, the two robots should always avoid the obstacles (grey shaded areas).
In particular, such a task can be described by the following formula:
\begin{align}
    \phi \coloneqq \mathbf{G}_{[4, 6]} \pi^{\mu_1} \wedge \mathbf{G}_{[9, 13]} \pi^{\mu_2} \wedge \mathbf{F}_{[16, 18]}\mathbf{G}_{[0, 2]} (\pi^{\mu_3} \wedge \pi^{\mu_4})
    \wedge \mathbf{G}_{[0, 20]} (\pi^{\mu_{close}} \wedge \pi^{\mu_{obs1}}\wedge \pi^{\mu_{obs2}}), \nonumber
\end{align}
where $\mu_1 \coloneqq \min (P_x^2, 2-P_x^2, P_y^2-4, 6-P_y^2)$, $\mu_2 \coloneqq \min (P_x^2-3.5, 6.5-P_x^2, P_y^2-8, 10-P_y^2)$, $\mu_3 \coloneqq \min (P_x^2-8.5, 10-P_x^2, P_y^2, 2-P_y^2)$ and $\mu_4 \coloneqq \min (p_x^1-7, 8.5-p_x^1, p_y^1, 2-P_y^1)$ represent the tasks for $R_2$ in the left blue region, $R_2$ in the top blue region, $R_2$ in the bottom right blue region, and $R_1$ in the bottom red region, respectively.  Here,
$\mu_{close} \coloneqq D - \min (P_x^2-p_x^1, p_x^1-P_x^2, P_y^2-p_y^1, p_y^1-P_y^2)$ is the predicate regarding the distance requirement between the two robots where we set $D=2$. The predicates
$\mu_{obs1} \coloneqq \min ( \max(1.6 - p_x^1, p_x^1 - 2.6, 2 - p_y^1, p_y^1 - 3), \max(8.3 - p_x^1, p_x^1 - 9.3, 6.5 - p_y^1, p_y^1 - 7.5), \max(5.7 - p_x^1, p_x^1 - 6.7, 2.7 - p_y^1, p_y^1 - 3.7))$ and $\mu_{obs2} \coloneqq \min ( \max(1.6 - P_x^2, P_x^2 - 2.6, 2 - P_y^2, P_y^2 - 3), \max(8.3 - P_x^2, P_x^2 - 9.3, 6.5 - P_y^2, P_y^2 - 7.5), \max(5.7 - P_x^2, P_x^2 - 6.7, 2.7 - P_y^2, P_y^2 - 3.7))$ describe the tasks for obstacle avoidance.
From an individual robot perspective, we can decompose $\phi$ into two task for each of the two robots as follows
\begin{align*}
    & \phi_{R_1} \coloneqq \mathbf{F}_{[16, 18]}\mathbf{G}_{[0, 2]} \pi^{\mu_4}  \wedge \mathbf{G}_{[0, 20]} (\pi^{\mu_{close}} \wedge \pi^{\mu_{obs1}}), \\
    & \phi_{R_2} \coloneqq \mathbf{G}_{[4, 6]} \pi^{\mu_1} \wedge \mathbf{G}_{[9, 13]} \pi^{\mu_2} \wedge \mathbf{F}_{[16, 18]}\mathbf{G}_{[0, 2]} \pi^{\mu_3} \wedge \mathbf{G}_{[0, 20]} \pi^{\mu_{obs2}}.
\end{align*}
As $R_2$ will lead the task, its behavior will not be influenced by $R_1$. On the other hand, $R_1$ is responsible for the task $\mu_{close}$ which means it should always track the behavior of $R_2$.
Regarding the execution of the task, we only enforce $R_1$ to achieve $\phi_{R_1}$ instead of $\phi$, which makes sense in the leader-follower setting of this case study.

\textbf{Motion of uncontrollable agent. }
The motion of the uncontrollable agent $R_2$ is described as follows. We set its dynamical system as $y_{k+1} = f(y_{k}, u_k^y) + \omega_k$ where $f(y_{k}, u_k^y)$ describes the same double integrator dynamics as for $R_1$, and $\omega_k \in \mathcal{W}$ is a uniformly distributed disturbance from the set $\mathcal{W} = ([-0.15,0.15] \times [0,0])^2$.
We compute an closed-loop control sequence for $R_2$ under the specification $\phi_{R_2}$ using a standard MIP encoding \cite{raman2014model}. 

\textbf{Data collection.} We collected $2000$ trajectories of the uncontrollable agent. We split the data  into training, calibration, and test datasets with sizes $|D_{train}|=500, |D_{cal}|=500$ and $|D_{test}|=1000$, respectively. We again trained an LSTM on $D_{train}$ for trajectory prediction.

\textbf{Prediction regions. } Fig.~\ref{fig:case2_his} shows histograms of the nonconformity scores $R_{OL}^{(j)}$ and $R_{CL}^{(j)}$ evaluated on $D_{cal}$. Based on these nonconformity scores, we have that $C_{OL} = 0.876$ and $C_{CL} = 0.997$ by using $\delta = 0.15$.
Next, we empirically validate the correctness of the prediction regions by checking whether or not \eqref{eq:cp_open_} and \eqref{eq:cp_closed} hold on $D_{test}$ by the same two experiments as the temperature control case. 
Fig.~\ref{fig:case2_his_emp} shows the histogram over the result of the ratios in the first experiment, and we can observe that the result achieves the desired coverage $1-\delta=0.85$.
In the second experiment, we obtain a ratio of 0.854 for $R_{OL}^{(j)} \leq C_{OL}$ and a ratio of 0.894 for $R_{CL}^{(j)} \leq C_{CL}$, respectively, which also empirically confirms \eqref{eq:cp_open_} and \eqref{eq:cp_closed}.
For one of these test trajectories, Fig.~\ref{fig:case2_trajectory} shows the prediction regions defined by $\hat{Y}_{\tau|k}$ and $C_{\tau|k}$ for $\tau > k$ where $k=0$, $k=6$, $k=12$ and $k=18$, respectively.

\textbf{Predictive control synthesis. } 
As in our first experiment, we run the open-loop and the closed-loop controller with the qualitative encoding $1000$ times with data from $D_{test}$. The cost function encodes the trade off between the sum of the squared control inputs and velocities as $J \coloneqq \sum_{k=0}^{20} (0.97||u_k||^2 + 0.03(v_x^1)^2 + 0.03(v_y^1)^2)$. In the open-loop control case, the problem is feasible for 1000 times, where the average computation time, average robustness, and task satisfaction rate are 0.70s, 0.06490, and 1, respectively.
In the closed-loop control case, the problem is feasible at the initial time $k=0$ for $1000$ times, with recursive feasibility in all cases, where the average total computation time, average robustness, and  task satisfaction rate are 14.80s, 0.06489, and 1, respectively.
The closed-loop case has again a slightly lower robustness for the same reasons as in the first experiment.
Fig.~\ref{fig:case2_trajectory} shows one of the results of the proposed control framework, where we can see that robot $R_1$ will use the prediction regions of $R_2$ to compute a control input that results in task satisfaction. Fig.~\ref{fig:case2_trajectory}(a) specifically presents the result of the open loop controller while Figs. \ref{fig:case2_trajectory}(b)(c)(d) present the result of the closed loop controller at times $k=6, 12, 18$, respectively.

\section{Conclusion}\label{sec:con}
In this paper, we presented a control framework for signal temporal logic (STL) control synthesis  where the STL specification is defined over the dynamical system and uncontrollable dynamic agents. We use trajectory predictors and conformal prediction to provide prediction regions for the uncontrollable agents that are valid with high probability.
Then, we formulated predictive open-loop control laws that consider the worst case realization of the uncontrollable agents within the prediction region. Specifically, we proposed a mixed integer program (MIP) that results in the probabilistic satisfaction of the STL specification, and we presented an equivalent MIP program based on the KKT conditions of the original MIP program to obtain more efficient solutions.
Finally, we presented a probabilistic recursively feasible closed-loop control framework that can provide high probability guarantee for the task satisfaction.
We provided two case studies for temperature control and robot motion planning, and we demonstrated the correctness and efficiency of our control framework.

\bibliographystyle{ACM-Reference-Format}
\bibliography{STL}

\end{document}

%% file: floor_plan.tex
\usetikzlibrary {patterns.meta}

\def\unit{1}
\def\height{2.5}
\def\heightdoora{1.05}
\def\heightdoorb{1.45}

\def\heighttexta{1.35}
\def\heighttextb{1.15}

\begin{tikzpicture}

    \draw[thick] (2.2*\unit, 0*\unit) -- (0*\unit, 0*\unit) -- (0*\unit, \height*\unit)  -- (5*\unit, \height*\unit) -- (5*\unit, 0*\unit)  -- (2.8*\unit, 0*\unit);

    \draw[thick] (1.7*\unit, 0*\unit) -- (1.7*\unit, \heightdoora*\unit);
    \draw[thick] (1.7*\unit, \heightdoorb*\unit) -- (1.7*\unit, \height*\unit);
    \draw[thick] (3.3*\unit, 0*\unit) -- (3.3*\unit, \heightdoora*\unit);
    \draw[thick] (3.3*\unit, \heightdoorb*\unit) -- (3.3*\unit, \height*\unit);

    \draw[blue, densely dashed] (1.7*\unit, \heightdoora*\unit) -- (1.7*\unit, \heightdoorb*\unit);
    \draw[blue, densely dashed] (3.3*\unit, \heightdoora*\unit) -- (3.3*\unit, \heightdoorb*\unit);
    \draw[blue, densely dashed] (2.2*\unit, 0*\unit) -- (2.8*\unit, 0*\unit);
	
    \draw (0.8*\unit, \heighttexta*\unit) node {\tiny{$Y_{r2}$: Room 2}};
    \draw (0.8*\unit, \heighttextb*\unit) node {\tiny{temp.}};
    \draw (4.2*\unit, \heighttexta*\unit) node {\tiny{$Y_{r3}$: Room 3}};
    \draw (4.2*\unit, \heighttextb*\unit) node {\tiny{temp.}};
    \draw (2.5*\unit, \heighttexta*\unit) node {\tiny{$x$: Room 1}};
    \draw (2.5*\unit, \heighttextb*\unit) node {\tiny{temp.}};


    




	





\end{tikzpicture}